\documentclass{article}%
\usepackage{amsfonts}
\usepackage{amsmath}
\usepackage{amssymb}
\usepackage{graphicx}
\usepackage{geometry}
\usepackage{color}%
\providecommand{\U}[1]{\protect\rule{.1in}{.1in}}
\newtheorem{theorem}{Theorem}

\newtheorem{lemma}[theorem]{Lemma}

\newtheorem{proposition}[theorem]{Proposition}

\newenvironment{proof}[1][Proof]{\noindent\textbf{#1.} }{\ \rule{0.5em}{0.5em}}
\begin{document}

\title{\textbf{Bianchi spaces and their $3$-dimensional isometries as $S$-expansions
of $2$-dimensional isometries}}
\author{Ricardo Caroca$^{1}$, Igor Kondrashuk$^{2}$, Nelson Merino$^{3,4}$, Felip
Nadal$^{5,4}$\\$^{1}${\small Departamento de Matem\'{a}tica y F\'{\i}sica Aplicadas,
Universidad Cat\'{o}lica de la Sant\'{\i}sima }\\{\small Concepci\'{o}n, Alonso de Rivera 2850, Concepci\'{o}n, Chile}\\$^{2}${\small Departamento de Ciencias B\'{a}sicas, Univerdidad del
B\'{\i}o-B\'{\i}o, }\\{\small Campus Fernando May, Casilla 447, Chill\'{a}n, Chile}\\$^{3}$ {\small Departamento de F\'{\i}sica, Universidad de Concepci\'{o}n,
Casilla 160-C, Concepci\'{o}n, Chile}\\$^{4}${\small Dipartimento di Fisica, Politecnico di Torino, }\\{\small Corso Duca degli Abruzzi, 24, I-10129 Torino, Italia}\\$^{5}${\small IFIC (Instituto de F\'{\i}sica Corpuscular) }\\{\small Edificio Institutos de Investigaci\'{o}n. c/ Catedr\'{a}tico Jos\'{e}
Beltr\'{a}n, 2. E-46980 Paterna, Spain}}
\maketitle

\begin{abstract}
In this paper  we show that some $3$-dimensional isometry
algebras, specifically those of type I, II, III and V ( according
Bianchi's classification), can be obtained as expansions of the isometries in
$2$ dimensions. It is shown that in general more than one semigroup will lead
to the same  result. It is impossible to obtain the algebras of
type IV, VI-IX as an expansion from the isometry algebras in $2$ dimensions.
This means that the first set of algebras has properties that can be obtained
from isometries in $2$ dimensions while the second set  has properties that
are in some sense intrinsic in $3$ dimensions. All the results are checked
with computer programs. This procedure can be
generalized to higher dimensions, which could be useful for diverse
physical applications.

\end{abstract}

\newpage

\section{\textbf{Introduction}}

Expansions of Lie algebras are a generalization of the contraction method
 which were introduced some years ago in refs. \cite{hs},
\cite{aipv1}, \cite{aipv2}, \cite{aipv3}, \cite{irs}, \cite{irs1},
\cite{irs2}. These methods, in particular the $S$-expansion procedure
developed in ref. \cite{irs}, are powerful tools to find \textit{non-trivial}%
\footnote{By \textit{non-trivial relations} we mean that these mechanisms of
contractions and expansions allow us to obtain some Lie algebras starting with
other algebras that have completely different properties. Also, the original
algebra is not necessarily (could be in specific cases) contained as a
subalgebra of the algebra obtained by these processes.} relations between different Lie
algebras which are a very interesting problem from both, physical and
mathematical points of view. In fact, many physical applications
have been found in this context (see, for example, refs. \cite{aipv1},
\cite{irs1}, \cite{irs2}, \cite{K15}, \cite{K11}, \cite{K12} and \cite{K14}).

It is the aim of this paper to show that the $S$-expansion method permits us to
obtain some types of $3$-dimensional isometries from the
$2$-dimensional isometries. We present a complete study about the possibility of
finding non-trivial relations between  $2$- and $3$-dimensional
isometry Lie algebras. Even when these isometries are well known in the
literature (see \cite{bian}), the non-trivial relations
we find between $2$- and $3$-dimensional isometry algebras are new
and interesting results. In fact, we identify the $3$-dimensional algebras
that can be obtained from the $2$-dimensional algebras by means of an
$S$-expansion and describe explicitly how to do it. For  other
$3$-dimensional algebras we show that it is impossible to obtain them
from the information of the $2$-dimensional algebras, and in some sense the
information that they contain is intrinsic to $3$ dimensions. A
possible generalization of this procedure to higher dimensions can be useful
in some physical applications (see section \ref{Com}).

This paper is organized as follows: In section \ref{desc} we
present a brief technical description of the basic ingredients that we are
going to use along this work. In section \ref{Sexp} we will review
some aspects of the $S$-expansion procedure while in section
\ref{history} we summarize the history about the enumeration and
characterization of finite semigroups existing for each order $n=1,...,9$. 
In section \ref{resred} we introduce a general kind of expansions that 
will be made in next sections and  in section \ref{clasification} we briefly describe 
the Bianchi classification of isometries in $2$- and $3$-dimensional spaces. 
In section \ref{Bianchi_related} it is shown in an instructive way how some types of
isometries are related with $2$-dimensional isometries using known semigroups
and also by introducing other semigroups that have not been used earlier in
the applications of the $S$-expansion procedure. In section \ref{Bsum} we
briefly summarize the results obtained by this iterative procedure. In section
\ref{Bianchi_no_related} it is shown why is it not possible to obtain, by
expansions, the other 3-dimensional isometries from the 2-dimensional algebras.
Finally in section \ref{check} we check the results using
computer programs and solve the problem entirely.

\section{Technical description}

\label{desc}

\subsection{\textbf{The S-expansion procedure}}

\label{Sexp}

In this section we briefly describe the general abelian semigroup expansion
procedure ($S$-expansion for short). We refer  to
ref.~\cite{irs} for further details.

Consider a Lie algebra $\mathcal{G}$ and a finite abelian semigroup
$S=\left\{  \lambda_{\alpha}\right\}  $. According to Theorem~3.1 from
ref.~\cite{irs}, the Cartesian product
\begin{equation}
\mathcal{G}_{S}\text{\ }\mathcal{=}\text{\ }S\times\mathcal{G}  \text{,}%
\label{s1}%
\end{equation}
is also a Lie algebra. The elements of this expanded algebra are denoted by%
\begin{equation}
X_{\left(  i,\alpha\right)  }=X_{i}\lambda_{\alpha}\label{z2}%
\end{equation}
where the product is understood as a direct product of matrix
representations of the generators $X_{i}$ of $\mathcal{G}$\ and the elements
$\lambda_{\alpha}$ of the semigroup $S$. The Lie product in $\mathcal{G}_{S}$ is
defined as%
\begin{equation}
\left[  X_{\left(  i,\alpha\right)  },X_{\left(  j,\beta\right)  }\right]
=\lambda_{\alpha}\lambda_{\beta}\left[  X_{i},X_{j}\right]  \label{s2'}%
\end{equation}
The set (\ref{s1}) with the composition law (\ref{s2'}) is called a $S$-expanded Lie algebra. \ 

In a nutshell, the $S$-expansion method can be seen as the natural
generalization of the In\"{o}n\"{u}-Wigner contraction, where instead of
multiplying the generators by a numerical parameter, we multiply the
generators by the elements of an Abelian semigroup, for more detail see
\cite{irs}.

\subsection{\textbf{Resonant subalgebra and }$0_{S}$-\textbf{reduced algebra}}

\label{resred}

As shown in \cite{irs} smaller algebras can be extracted from the above
expanded algebra: the resonant subalgebra and the reduced algebra.  Their
existence depends on certain conditions expressed by the Eqs. (23)
and (34) of ref. \cite{irs}.

In fact the original algebra will be one of the $2$-dimensional isometry
algebras (\ref{2dimalg}-\ref{2dimalg_dos}) both having the subspace structure $\mathcal{G}%
=V_{0}\oplus V_{1}$ given by%
\begin{align}
\left[  V_{0},V_{0}\right]   &  \subset V_{0}\ \\
\left[  V_{0},V_{1}\right]   &  \subset V_{1}\label{subs_structure}\\
\left[  V_{1},V_{1}\right]   &  \subset V_{0}\nonumber
\end{align}
where $V_{0}$, $V_{1}$ are respectively generated by $X_{2}$ and $X_{1}$. On
the other hand, the semigroups we are going to construct will
possess a resonant decomposition, i.e, must  be of the form $S=S_{0}\cup
S_{1}$ where%

\begin{align}
S_{0}\times S_{0} &  \subset S_{0}\nonumber\\
S_{0}\times S_{1} &  \subset S_{1}\nonumber\\
S_{1}\times S_{1} &  \subset S_{0}\label{res_con}%
\end{align}
Then, according to Theorem 4.2 of ref. \cite{irs}, the resonant subalgebra is of of the form%
\begin{equation}
\mathcal{G}_{S,R}=\left(  S_{0}\times V_{0}\right)  \oplus\left(  S_{1}\times
V_{1}\right)  \label{re_algebra}%
\end{equation}
Note that Eq. (\ref{res_con}) is a particular case of  Eq. (34) of Ref. (\cite{irs}).

An even smaller algebra can be obtained when there is a zero element in the
semigroup, i.e., an element $0_{S}\in S$ such that, for all $\lambda_{\alpha
}\in S$, $0_{S}\lambda_{\alpha}=0_{S}$. When this is the case, the whole
$0_{S}\times \mathcal{G}$ sector can be removed from the resonant subalgebra by imposing
$0_{S}\times \mathcal{G}=0$ (see Definition 3.3 from ref.~\cite{irs}). The resulting
algebra\ continues to be a Lie algebra and here it will be denoted by
$\mathcal{G}_{S,R}^{\text{red}}$.

\subsection{\textbf{Finite semigroups programming}}

\label{history}

The numbers of finite non-isomorphic semigroups of order $n$ are given in the
following table:%

\begin{equation}%
\begin{tabular}
[c]{|l|l|l}\cline{1-2}%
order & $Q=\ $\# semigroups & \\\cline{1-2}%
1 & 1 & \\\cline{1-2}%
2 & 4 & \\\cline{1-2}%
3 & 18 & \\\cline{1-2}%
4 & 126 & [Forsythe '54]\\\cline{1-2}%
5 & 1,160 & [Motzkin, Selfridge '55]\\\cline{1-2}%
6 & 15,973 & [Plemmons '66]\\\cline{1-2}%
7 & 836,021 & [Jurgensen, Wick '76]\\\cline{1-2}%
8 & 1,843,120,128 & [Satoh, Yama, Tokizawa '94]\\\cline{1-2}%
9 & \textbf{52,989,400,714,478} & [Distler, Kelsey, Mitchell '09]\\\cline{1-2}%
\end{tabular}
\label{hist}
\end{equation}
All the  semigroups of order 4 have been classified by Forsythe in Ref.
\cite{n4}, of order 5 by Motzkin and Selfridge in Ref. \cite{n5}, of order 6
by Plemmons in Ref. \cite{n6-1,n6-2,n6-3}, of order 7 by J\"{u}rgensen and
Wick in Ref. \cite{n7}, and of order 8 by Satoh, Yama and Tokizawa in Ref.
\cite{n8}, and monoids and semigroups of order 9 by Distler and Kelsey in Ref.
\cite{n9-1,n9-2} and by Distler and Mitchell in Ref. \cite{n9-3}. Also, for
semigroups of order 9 the result can be found in Ref. \cite{n9-4}.

As shown in the table the problem of enumerating the all non-isomorphic finite
semigroups of a certain order is a non-trivial problem. In fact, the number
$Q$ of semigroups increases very quickly with the order of the semigroup.

In ref. \cite{Plemmons} a set of algorithms is given that permit us to make
certain calculations with finite semigroups. The first program, \textit{gen.f}%
, gives all the non-isomorphic semigroups of order $n$ for $n=1,2,...,8$%
\footnote{The order $n=9$  is non trivial and the algorithms of
the mentioned reference fails. This non-trivial
problem was solved in 2009 by Andreas Distler, Tom Kelsey \& James Mitchell.
However, in this paper we are going to consider calculations with semigroups of at
most the order $4$.}. The input is the order, $n$, of the semigroups we
want to obtain and the \textit{output} is a list of all the non-isomorphic
semigroups that exist in this order. In this work, the elements of the
semigroup are labeled by $\lambda_{\alpha}$ with $\alpha=1,...,n$ and each
semigroup will be denoted by $S_{\left(  n\right)  }^{a}$ where the
super-index $a=1,...,Q$ identifies the specific semigroup of order $n$. The
second program in \cite{Plemmons}, \textit{com.f}, takes as \textit{input} one
of the mentioned list for a certain order, picks up just the symmetric tables
and generates another list with all the abelian semigroups. For example for
$n=3$ the elements of the semigroup are labeled by $\lambda_{1}$, $\lambda
_{2}$ and $\lambda_{3}$ and the program \textit{com.f} gives the following
list of semigroups:%

\[%
\begin{tabular}
[c]{l|lll}%
$S_{\left(  3\right)  }^{1}$ & $\lambda_{1}$ & $\lambda_{2}$ & $\lambda_{3}%
$\\\hline
$\lambda_{1}$ & $\lambda_{1}$ & $\lambda_{1}$ & $\lambda_{1}$\\
$\lambda_{2}$ & $\lambda_{1}$ & $\lambda_{1}$ & $\lambda_{1}$\\
$\lambda_{3}$ & $\lambda_{1}$ & $\lambda_{1}$ & $\lambda_{1}$%
\end{tabular}
\ \ \text{, }%
\begin{tabular}
[c]{l|lll}%
$S_{\left(  3\right)  }^{2}$ & $\lambda_{1}$ & $\lambda_{2}$ & $\lambda_{3}%
$\\\hline
$\lambda_{1}$ & $\lambda_{1}$ & $\lambda_{1}$ & $\lambda_{1}$\\
$\lambda_{2}$ & $\lambda_{1}$ & $\lambda_{1}$ & $\lambda_{1}$\\
$\lambda_{3}$ & $\lambda_{1}$ & $\lambda_{1}$ & $\lambda_{2}$%
\end{tabular}
\ \ \text{, }%
\begin{tabular}
[c]{l|lll}%
$S_{\left(  3\right)  }^{3}$ & $\lambda_{1}$ & $\lambda_{2}$ & $\lambda_{3}%
$\\\hline
$\lambda_{1}$ & $\lambda_{1}$ & $\lambda_{1}$ & $\lambda_{1}$\\
$\lambda_{2}$ & $\lambda_{1}$ & $\lambda_{1}$ & $\lambda_{1}$\\
$\lambda_{3}$ & $\lambda_{1}$ & $\lambda_{1}$ & $\lambda_{3}$%
\end{tabular}
\ \ \text{, }%
\begin{tabular}
[c]{l|lll}%
$S_{\left(  3\right)  }^{6}$ & $\lambda_{1}$ & $\lambda_{2}$ & $\lambda_{3}%
$\\\hline
$\lambda_{1}$ & $\lambda_{1}$ & $\lambda_{1}$ & $\lambda_{1}$\\
$\lambda_{2}$ & $\lambda_{1}$ & $\lambda_{1}$ & $\lambda_{2}$\\
$\lambda_{3}$ & $\lambda_{1}$ & $\lambda_{2}$ & $\lambda_{3}$%
\end{tabular}
\ \ \text{, }%
\]

\[%
\begin{tabular}
[c]{l|lll}%
$S_{\left(  3\right)  }^{7}$ & $\lambda_{1}$ & $\lambda_{2}$ & $\lambda_{3}%
$\\\hline
$\lambda_{1}$ & $\lambda_{1}$ & $\lambda_{1}$ & $\lambda_{1}$\\
$\lambda_{2}$ & $\lambda_{1}$ & $\lambda_{2}$ & $\lambda_{1}$\\
$\lambda_{3}$ & $\lambda_{1}$ & $\lambda_{1}$ & $\lambda_{3}$%
\end{tabular}
\ \text{, }%
\begin{tabular}
[c]{l|lll}%
$S_{\left(  3\right)  }^{9}$ & $\lambda_{1}$ & $\lambda_{2}$ & $\lambda_{3}%
$\\\hline
$\lambda_{1}$ & $\lambda_{1}$ & $\lambda_{1}$ & $\lambda_{1}$\\
$\lambda_{2}$ & $\lambda_{1}$ & $\lambda_{2}$ & $\lambda_{2}$\\
$\lambda_{3}$ & $\lambda_{1}$ & $\lambda_{2}$ & $\lambda_{2}$%
\end{tabular}
\ \text{, }%
\begin{tabular}
[c]{l|lll}%
$S_{\left(  3\right)  }^{10}$ & $\lambda_{1}$ & $\lambda_{2}$ & $\lambda_{3}%
$\\\hline
$\lambda_{1}$ & $\lambda_{1}$ & $\lambda_{1}$ & $\lambda_{1}$\\
$\lambda_{2}$ & $\lambda_{1}$ & $\lambda_{2}$ & $\lambda_{2}$\\
$\lambda_{3}$ & $\lambda_{1}$ & $\lambda_{2}$ & $\lambda_{3}$%
\end{tabular}
\ \text{, }%
\begin{tabular}
[c]{l|lll}%
$S_{\left(  3\right)  }^{12}$ & $\lambda_{1}$ & $\lambda_{2}$ & $\lambda_{3}%
$\\\hline
$\lambda_{1}$ & $\lambda_{1}$ & $\lambda_{1}$ & $\lambda_{1}$\\
$\lambda_{2}$ & $\lambda_{1}$ & $\lambda_{2}$ & $\lambda_{3}$\\
$\lambda_{3}$ & $\lambda_{1}$ & $\lambda_{3}$ & $\lambda_{2}$%
\end{tabular}
\ \text{, }%
\]

\[%
\begin{tabular}
[c]{l|lll}%
$S_{\left(  3\right)  }^{15}$ & $\lambda_{1}$ & $\lambda_{2}$ & $\lambda_{3}%
$\\\hline
$\lambda_{1}$ & $\lambda_{1}$ & $\lambda_{1}$ & $\lambda_{3}$\\
$\lambda_{2}$ & $\lambda_{1}$ & $\lambda_{1}$ & $\lambda_{3}$\\
$\lambda_{3}$ & $\lambda_{3}$ & $\lambda_{3}$ & $\lambda_{1}$%
\end{tabular}
\ \text{, }%
\begin{tabular}
[c]{l|lll}%
$S_{\left(  3\right)  }^{16}$ & $\lambda_{1}$ & $\lambda_{2}$ & $\lambda_{3}%
$\\\hline
$\lambda_{1}$ & $\lambda_{1}$ & $\lambda_{1}$ & $\lambda_{3}$\\
$\lambda_{2}$ & $\lambda_{1}$ & $\lambda_{2}$ & $\lambda_{3}$\\
$\lambda_{3}$ & $\lambda_{3}$ & $\lambda_{3}$ & $\lambda_{1}$%
\end{tabular}
\ \text{, }%
\begin{tabular}
[c]{l|lll}%
$S_{\left(  3\right)  }^{17}$ & $\lambda_{1}$ & $\lambda_{2}$ & $\lambda_{3}%
$\\\hline
$\lambda_{1}$ & $\lambda_{1}$ & $\lambda_{2}$ & $\lambda_{2}$\\
$\lambda_{2}$ & $\lambda_{2}$ & $\lambda_{1}$ & $\lambda_{1}$\\
$\lambda_{3}$ & $\lambda_{2}$ & $\lambda_{1}$ & $\lambda_{1}$%
\end{tabular}
\ \text{, }%
\begin{tabular}
[c]{l|lll}%
$S_{\left(  3\right)  }^{18}$ & $\lambda_{1}$ & $\lambda_{2}$ & $\lambda_{3}%
$\\\hline
$\lambda_{1}$ & $\lambda_{1}$ & $\lambda_{2}$ & $\lambda_{3}$\\
$\lambda_{2}$ & $\lambda_{2}$ & $\lambda_{3}$ & $\lambda_{1}$\\
$\lambda_{3}$ & $\lambda_{3}$ & $\lambda_{1}$ & $\lambda_{2}$%
\end{tabular}
\ \text{, }%
\]

So in general the program \textit{com.f} of \cite{Plemmons} gives a list of
tables of the all abelian non-isomorphic semigroups of a certain
order\footnote{As the number of non-isomorphic semigroups increases very
quickly with the order $n$ (see table (\ref{hist})), the mentioned lists are
very large for higher orders. Note also that the semigroups $S_{\left(
3\right)  }^{4}$, $S_{\left(  3\right)  }^{5}$, $S_{\left(  3\right)  }^{8}$,
$S_{\left(  3\right)  }^{11}$, $S_{\left(  3\right)  }^{13}$ and $S_{\left(
3\right)  }^{14}$ are not given in the list for $n=3$ because they are not
abelian  (non-commutative).} (up to order $8$).

In section \ref{check}  we show that each semigroup that we
construct in this paper, to make the relations between $2$- and $3$- dimensional
isometries, are isomorphic to one of the semigroups of the lists given by
\cite{Plemmons}. This is a way to check that the iterative procedure (that we
shall present in sections \ref{Bianchi_related}, \ref{Bianchi_no_related}) to
find semigroups with zero elements and resonant decompositions is working well.

\subsection{Bianchi classification}

\label{clasification}

In ref. \cite{bian} Bianchi proposed a procedure to classify the
$3$-dimensional spaces that admit a $3$-dimensional isometry.
He had shown how to represent the generators as Killing vectors and how to get the
corresponding metrics,  and formulated Bianchi's theorem. 

Here we are interested  in study  of the possibility of relating, by
means of expansions, $2$- and $3$-dimensional isometry algebras. The $2$-dimensional algebras 
are given by
\begin{align}
\left[  X_{1},X_{2}\right]   &  =0\text{ or}\label{2dimalg}\\
\left[  X_{1},X_{2}\right]   &  =X_{1}\label{2dimalg_dos}%
\end{align}
The $3-$dimensional algebras are given in the following table,%

\begin{equation}%
\begin{tabular}
[c]{|l|l|}\hline
\textbf{Group} & \textbf{Algebra}\\\hline
type I & $\left[  X_{1},X_{2}\right]  =\left[  X_{1},X_{3}\right]  =\left[
X_{2},X_{3}\right]  =0$\\\hline
type II & $\left[  X_{1},X_{2}\right]  =\left[  X_{1},X_{3}\right]
=0,\ \ \ \left[  X_{2},X_{3}\right]  =X_{1}$\\\hline
type III & $\left[  X_{1},X_{2}\right]  =\left[  X_{2},X_{3}\right]
=0,\ \ \ \left[  X_{1},X_{3}\right]  =X_{1}$\\\hline
type IV & $\left[  X_{1},X_{2}\right]  =0,\ \ \ \left[  X_{1},X_{3}\right]
=X_{1},\ \ \ \left[  X_{2},X_{3}\right]  =X_{1}+X_{2}$\\\hline
type V & $\left[  X_{1},X_{2}\right]  =0,\ \ \ \left[  X_{1},X_{3}\right]
=X_{1},\ \ \ \left[  X_{2},X_{3}\right]  =X_{2}$\\\hline
type VI &
\begin{tabular}
[c]{l}%
$\left[  X_{1},X_{2}\right]  =0,\ \ \ \left[  X_{1},X_{3}\right]
=X_{1},\ \ \ \left[  X_{2},X_{3}\right]  =hX_{2},$\\
where $h\neq0,1$%
\end{tabular}
$\ .$\\\hline
type VII$_{1}$ & $\left[  X_{1},X_{2}\right]  =0,\ \ \ \left[  X_{1}%
,X_{3}\right]  =X_{2},\ \ \ \left[  X_{2},X_{3}\right]  =-X_{1}$\\\hline
type VII$_{2}$ &
\begin{tabular}
[c]{l}%
$\left[  X_{1},X_{2}\right]  =0,\ \ \ \left[  X_{1},X_{3}\right]
=X_{2},\ \ \ \left[  X_{2},X_{3}\right]  =-X_{1}+hX_{2},$\\
where$\ h\neq0\ (0<h<2).$%
\end{tabular}
\\\hline
type VIII & $\left[  X_{1},X_{2}\right]  =X_{1},\ \ \ \left[  X_{1}%
,X_{3}\right]  =2X_{2},\ \ \ \left[  X_{2},X_{3}\right]  =X_{3}$\\\hline
type IX & $\left[  X_{1},X_{2}\right]  =X_{3},\ \ \ \left[  X_{2}%
,X_{3}\right]  =X_{1},\ \ \ \left[  X_{3},X_{1}\right]  =X_{2}$\\\hline
\end{tabular}
\ \ \ \label{Bianchi}%
\end{equation}

\section{\textbf{Main proposition}}

The whole paper is dedicated to demonstrate that the following proposition is valid.

\begin{proposition}
Some but not all of the Bianchi algebras can be obtained as $S$-expansions of the
2-dimensional algebras (\ref{2dimalg}-\ref{2dimalg_dos}).
\end{proposition}

This proposition will be proven by performing the $S$-expansion of
(\ref{2dimalg}-\ref{2dimalg_dos}) and by applying two known procedures which
permit to extract a smaller algebra from the expanded algebra. As we mention
in section \ref{resred}, those are the construction of resonant
subalgebra and the $0_S$-reduction of the resonant subalgebra.

\section{The Bianchi spaces related with $2$-dimensional isometries}

\label{Bianchi_related}

\subsection{The type III algebra}

\label{BtypeIII}

\begin{theorem}
An abelian semigroup of four elements $S=\left\{ \lambda_{1},\lambda_{2},\lambda_{3},\lambda_{4}\right\} $ 
with  decomposition
\begin{align}
S_{0}  &  =\left\{  \lambda_{2},\lambda_{3},\lambda_{4}\right\}
\label{decompIII}\\
S_{1}  &  =\left\{  \lambda_{1},\lambda_{4}\right\} \nonumber
\end{align}
compatible with resonance condition (\ref{res_con}) and possessing the following constraints 
in the multiplication table
\begin{equation}%
\begin{tabular}
[c]{l|llll}
& $\lambda_{1}$ & $\lambda_{2}$ & $\lambda_{3}$ & $\lambda_{4}$\\\hline
$\lambda_{1}$ &  & $\lambda_{4}$ & $\lambda_{1}$ & $\lambda_{4}$\\
$\lambda_{2}$ & $\lambda_{4}$ &  &  & $\lambda_{4}$\\
$\lambda_{3}$ & $\lambda_{1}$ &  &  & $\lambda_{4}$\\
$\lambda_{4}$ & $\lambda_{4}$ & $\lambda_{4}$ & $\lambda_{4}$ & $\lambda_{4}$%
\end{tabular}
\ \ \ \label{case3}%
\end{equation}
produces, after $0_S$-reduction of the resonant subalgebra of the $S$-expanded
algebra of the $2$-dimensional isometry (\ref{2dimalg_dos}), an algebra which
coincides with $3$-dimensional Bianchi type III isometry.
\end{theorem}

\begin{proof}
Let's begin with an unknown semigroup $S=\left\{  \lambda_{1},\lambda
_{2},\lambda_{3},\lambda_{4}\right\}  $ where we impose the following conditions:

$(i)$ $\lambda_{4}$ is a zero of the semigroup, so the table of
multiplication law is restricted to the form%
\[%
\begin{tabular}
[c]{l|llll}
& $\lambda_{1}$ & $\lambda_{2}$ & $\lambda_{3}$ & $\lambda_{4}$\\\hline
$\lambda_{1}$ &  &  &  & $\lambda_{4}$\\
$\lambda_{2}$ &  &  &  & $\lambda_{4}$\\
$\lambda_{3}$ &  &  &  & $\lambda_{4}$\\
$\lambda_{4}$ & $\lambda_{4}$ & $\lambda_{4}$ & $\lambda_{4}$ & $\lambda_{4}$%
\end{tabular}
\
\]
where the empty spaces must be filled in a way such that it is in fact an
abelian semigroup, i.e., closed, associative and commutative. In order to get
a smaller algebra\footnote{Note that we start with a $2$-dimensional algebra
and we are looking for a $3$-dimensional one. The $S$-expanded algebra with a
semigroup of order 4 is $8$-dimensional and the reduced is
$6$-dimensional. So to obtain a $3$-dimensional algebra we need to extract an
even smaller algebra. This is only possible by extracting a resonant
subalgebra.} we also demand

$(ii)$ that it contains a decomposition like that given in (\ref{decompIII}) which is resonant, i.e., that satisfies equation (\ref{res_con}).

Then, according to (\ref{re_algebra}), the resonant subalgebra of
$\mathcal{G}_{S,R} = S\times\mathcal{G}$ is given by\footnote{Note that we have an
abuse of notation here. The last equation must be read as: the resonant
subalgebra is generated by the set of generators that appears in the right
hand side.}%
\begin{align}
\mathcal{G}_{S,R}  &  =\left(  S_{0}\times V_{0}\right)  \oplus\left(
S_{1}\times V_{1}\right) \label{IIIR}\\
&  =\left\{  \lambda_{2}X_{2},\ \lambda_{3}X_{2},\ \lambda_{4}X_{2}\right\}
\oplus\left\{  \lambda_{1}X_{1},\ \lambda_{4}X_{1}\right\} \nonumber\\
&  =\left\{  \lambda_{2}X_{2},\ \lambda_{3}X_{2},\ \lambda_{4}X_{2}%
,\ \lambda_{1}X_{1},\ \lambda_{4}X_{1}\right\} \nonumber
\end{align}
Now we have to extract an even smaller algebra by means of a $0_{S}%
$-reduction. This is done by just taking off from (\ref{IIIR}) the elements
that contain the zero element, $\lambda_{4}$. Therefore, the
reduction of the resonant subalgebra is given by%
\[
\mathcal{G}_{S,R}^{\text{red}}=\left\{  \lambda_{2}X_{2},\ \lambda_{3}%
X_{2},\ \lambda_{1}X_{1}\right\}
\]
with the following commutation relations%
\begin{align*}
\left[  \lambda_{2}X_{2},\lambda_{3}X_{2}\right]   &  =0\\
\left[  \lambda_{2}X_{2},\lambda_{1}X_{1}\right]   &  =-\lambda_{1}\lambda
_{2}X_{1}\\
\left[  \lambda_{3}X_{2},\lambda_{1}X_{1}\right]   &  =-\lambda_{1}\lambda
_{3}X_{1}%
\end{align*}
In order for this algebra to be closed we should choose:

$(a)$ $\lambda_{1}\lambda_{2}=\lambda_{1}$ or $\lambda_{1}\lambda_{2}%
=\lambda_{4}$ and

$(b)$ $\lambda_{1}\lambda_{3}=\lambda_{1}$ or $\lambda_{1}\lambda_{3}%
=\lambda_{4}$.

So we are lead with four possibilities to construct a closed algebra:

$(i)$ $\lambda_{1}\lambda_{2}=\lambda_{1}\lambda_{3}=\lambda_{4}$,

$(ii)$ $\lambda_{1}\lambda_{2}=\lambda_{1}$ and $\lambda_{1}\lambda
_{3}=\lambda_{4}$,

$(iii)$ $\lambda_{1}\lambda_{2}=\lambda_{4}$ and $\lambda_{1}\lambda
_{3}=\lambda_{1}$,

$(iv)$ $\lambda_{1}\lambda_{2}=\lambda_{1}\lambda_{3}=\lambda_{1}$.

The case $(i)$ will lead to translations in $3$ dimensions, i.e., to the type
I algebra\footnote{This case will be analyzed later.}. It can be checked that
the case $(iv)$ it is not useful, because in this case the
multiplication law is non associative. On the other hand,  it can
be seen that both $(ii)$ and $(iii)$ will lead to the type III algebra. In
fact, in case $(ii)$ we have
\begin{align*}
\left[  \lambda_{2}X_{2},\lambda_{3}X_{2}\right]   &  =0\\
\left[  \lambda_{2}X_{2},\lambda_{1}X_{1}\right]   &  =-\lambda_{1}X_{1}\\
\left[  \lambda_{3}X_{2},\lambda_{1}X_{1}\right]   &  =-\lambda_{4}%
X_{1}=0\text{ (}\lambda_{4}\text{ is a  zero element)}%
\end{align*}
renaming the generators as%
\begin{align*}
Y_{1}  &  =\lambda_{1}X_{1}\\
Y_{2}  &  =\lambda_{3}X_{2}\\
Y_{3}  &  =\lambda_{2}X_{2}%
\end{align*}
we immediately recognize the type III algebra (see table (\ref{Bianchi})). In
case $(iii)$ we would have%
\begin{align*}
\left[  \lambda_{2}X_{2},\lambda_{3}X_{2}\right]   &  =0\\
\left[  \lambda_{2}X_{2},\lambda_{1}X_{1}\right]   &  =-\lambda_{4}%
X_{1}=0\text{ (}\lambda_{4}\text{ is a zero element)}\\
\left[  \lambda_{3}X_{2},\lambda_{1}X_{1}\right]   &  =-\lambda_{1}X_{1}%
\end{align*}
and we recover again the type III algebra by renaming the
generators as%
\begin{align*}
Y_{1}  &  =\lambda_{1}X_{1}\\
Y_{2}  &  =\lambda_{2}X_{2}\\
Y_{3}  &  =\lambda_{3}X_{2}%
\end{align*}

We choose to study case $(iii)$ to construct a semigroup that leads to Type III algebra, although the case $(ii)$ could also be studied to generate other semigroups leading to the same result. So,
the table describing the multiplication law case $(iii)$ is given by%
\[%
\begin{tabular}
[c]{l|llll}
& $\lambda_{1}$ & $\lambda_{2}$ & $\lambda_{3}$ & $\lambda_{4}$\\\hline
$\lambda_{1}$ &  & $\lambda_{4}$ & $\lambda_{1}$ & $\lambda_{4}$\\
$\lambda_{2}$ & $\lambda_{4}$ &  &  & $\lambda_{4}$\\
$\lambda_{3}$ & $\lambda_{1}$ &  &  & $\lambda_{4}$\\
$\lambda_{4}$ & $\lambda_{4}$ & $\lambda_{4}$ & $\lambda_{4}$ & $\lambda_{4}$%
\end{tabular}
\]
where the empty spaces must be filled in such a way that satisfy
associativity and the decomposition (\ref{decompIII}) satisfies the resonant
condition (\ref{res_con}).
\end{proof}

\begin{proposition}
There are semigroups that fit multiplication table (\ref{case3}) and resonant condition 
(\ref{decompIII}).
\end{proposition}

We give several examples of semigroups of this type.

\subsubsection{Semigroup $S_{K}^{(3)}$}

Consider the semigroup $S_{K}^{(n)}$, with the multiplication law defined by%
\begin{align}
\lambda_{\alpha}\lambda_{\beta}  &  =\lambda_{\min\left\{  \alpha
,\beta\right\}  }\text{, \ \ }\alpha+\beta>n\label{SIK}\\
\lambda_{\alpha}\lambda_{\beta}  &  =\lambda_{n+1}\text{, \ \ \ \ \ \ \ \ }%
\alpha+\beta\leq n\nonumber\\
\alpha,\beta &  =1,2,...,n\nonumber
\end{align}
It is directly seen that for $n=3$ the table of multiplication law%
\[%
\begin{tabular}
[c]{l|llll}
& $\lambda_{1}$ & $\lambda_{2}$ & $\lambda_{3}$ & $\lambda_{4}$\\\hline
$\lambda_{1}$ & $\lambda_{4}$ & $\lambda_{4}$ & $\lambda_{1}$ & $\lambda_{4}%
$\\
$\lambda_{2}$ & $\lambda_{4}$ & $\lambda_{2}$ & $\lambda_{2}$ & $\lambda_{4}%
$\\
$\lambda_{3}$ & $\lambda_{1}$ & $\lambda_{2}$ & $\lambda_{3}$ & $\lambda_{4}%
$\\
$\lambda_{4}$ & $\lambda_{4}$ & $\lambda_{4}$ & $\lambda_{4}$ & $\lambda_{4}$%
\end{tabular}
\ \ \
\]
fits with the form of the table (\ref{case3}), where we have filled those
empty spaces that appears in (\ref{case3}). Therefore, an expansion with the
semigroup $S_{K}^{(3)}$ reproduces the type III algebra after a reduction
of the resonant subalgebra.

\subsubsection{Semigroup $S_{N1}$}

Let's consider the following table of multiplication%
\[%
\begin{tabular}
[c]{l|llll}
& $\lambda_{1}$ & $\lambda_{2}$ & $\lambda_{3}$ & $\lambda_{4}$\\\hline
$\lambda_{1}$ & $\lambda_{4}$ & $\lambda_{4}$ & $\lambda_{1}$ & $\lambda_{4}%
$\\
$\lambda_{2}$ & $\lambda_{4}$ & $\lambda_{2}$ & $\lambda_{4}$ & $\lambda_{4}%
$\\
$\lambda_{3}$ & $\lambda_{1}$ & $\lambda_{4}$ & $\lambda_{3}$ & $\lambda_{4}%
$\\
$\lambda_{4}$ & $\lambda_{4}$ & $\lambda_{4}$ & $\lambda_{4}$ & $\lambda_{4}$%
\end{tabular}
\ \ \
\]
Here we have filled the empty spaces of (\ref{case3}) in another way leading
to a new semigroup that verifies the required properties. It can be
directly shown that this product is associative, satisfies the resonant
condition and fits the form of the table (\ref{case3}). The
proof is direct but a little tedious. Therefore this semigroup,
$S_{N1}$, also reproduces the type III algebra and it is not isomorphic to the
previous semigroup, $S_{K}^{(3)}$.

\subsubsection{The semigroup $S_{E}^{\left(  2\right)  }$, another way to
obtain the type III algebra}

In order to show that there are even other semigroups that lead
to the type III algebra we consider the semigroup $S_{E}^{\left(  2\right)  }$
introduced in ref. \cite{irs} for $n=2$. Its multiplication law is given by the following table%

\[%
\begin{tabular}
[c]{l|llll}
& $\lambda_{0}$ & $\lambda_{1}$ & $\lambda_{2}$ & $\lambda_{3}$\\\hline
$\lambda_{0}$ & $\lambda_{0}$ & $\lambda_{1}$ & $\lambda_{2}$ & $\lambda_{3}%
$\\
$\lambda_{1}$ & $\lambda_{1}$ & $\lambda_{2}$ & $\lambda_{3}$ & $\lambda_{3}%
$\\
$\lambda_{2}$ & $\lambda_{2}$ & $\lambda_{3}$ & $\lambda_{3}$ & $\lambda_{3}%
$\\
$\lambda_{3}$ & $\lambda_{3}$ & $\lambda_{3}$ & $\lambda_{3}$ & $\lambda_{3}$%
\end{tabular}
\]
and its resonant partition is%
\begin{align*}
S_{0}  &  =\left\{  \lambda_{0},\lambda_{2},\lambda_{3}\right\} \\
S_{1}  &  =\left\{  \lambda_{1},\lambda_{3}\right\}
\end{align*}

The $0_S$-reduction of the resonant subalgebra is given by%
\[
\mathcal{G}_{S,R}^{\text{red}}=\left\{  \lambda_{0}X_{2},\ \lambda_{2}%
X_{2},\ \lambda_{1}X_{1}\right\}
\]
with commutation relations%

\begin{align*}
\left[  \lambda_{0}X_{2},\lambda_{2}X_{2}\right]   &  =0\\
\left[  \lambda_{0}X_{2},\lambda_{1}X_{1}\right]   &  =-\lambda_{1}X_{1}\\
\left[  \lambda_{2}X_{2},\lambda_{1}X_{1}\right]   &  =0
\end{align*}
Renaming the generators as
\begin{align*}
Y_{1}  &  =\lambda_{1}X_{1}\\
Y_{2}  &  =\lambda_{2}X_{2}\\
Y_{3}  &  =\lambda_{0}X_{2}%
\end{align*}
we obtain again the type III algebra%
\[
\left[  Y_{1},Y_{2}\right]  =\left[  Y_{2},Y_{3}\right]  =0,\ \ \ \left[
Y_{1},Y_{3}\right]  =Y_{1}\text{.}%
\]

\subsection{\textbf{The type II and V algebras}}

The natural question here is: is it possible to generate other type of Bianchi
algebras from the isometries in $2$ dimensions? To answer this
question we will continue the procedure of section \ref{BtypeIII},
considering a semigroup $S=\left\{  \lambda_{1},\lambda_{2},\lambda
_{3},\lambda_{4}\right\}  $ where $\lambda_{4}$ is a zero element, but modify
resonant decomposition (\ref{decompIII}). The decomposition can be chosen in another way, 
as can be seen in the following

\begin{lemma}
A resonant decomposition
\begin{align}
S_{0}  &  =\left\{  \lambda_{2},\lambda_{4}\right\} \nonumber\\
S_{1}  &  =\left\{  \lambda_{1},\lambda_{3},\lambda_{4}\right\}
\label{newdes}%
\end{align}
satisfying resonant condition (\ref{res_con}) can produce,  
after $0_S$-reduction of resonant subalgebra of the
$S$-expanded algebra of 2-dimensional isometry (\ref{2dimalg_dos}), a 3-dimensional algebra. 
\end{lemma}

\begin{proof}
The reduction of the resonant subalgebra is given by%
\[
\mathcal{G}_{S,R}^{\text{red}}=\left\{  \lambda_{2}X_{2},\ \lambda_{1}%
X_{1},\ \lambda_{3}X_{1}\right\}
\]
and the commutation relations are given by%
\begin{align}
\left[  \lambda_{2}X_{2},\lambda_{1}X_{1}\right]   &  =-\lambda_{1}\lambda
_{2}X_{1}\label{seg}\\
\left[  \lambda_{2}X_{2},\lambda_{3}X_{1}\right]   &  =-\lambda_{2}\lambda
_{3}X_{1}\nonumber\\
\left[  \lambda_{1}X_{1},\lambda_{3}X_{1}\right]   &  =0\nonumber
\end{align}
Resonant condition  (\ref{res_con}) guarantees that (\ref{seg}) is a closed algebra.

\end{proof}

Here we have different possibilities in order to make this algebra closed.

\subsubsection{Type II algebra and the $S_{N2}$ semigroup}

\begin{theorem}
An abelian semigroup of four elements $S=\left\{ \lambda_{1},\lambda_{2},\lambda_{3},\lambda_{4}\right\} $ 
satisfying the conditions of Lemma 5 and possessing the following constraints in the multiplication table
\begin{equation}%
\begin{tabular}
[c]{l|llll}
& $\lambda_{1}$ & $\lambda_{2}$ & $\lambda_{3}$ & $\lambda_{4}$\\\hline
$\lambda_{1}$ &  & $\lambda_{3}$ &  & $\lambda_{4}$\\
$\lambda_{2}$ & $\lambda_{3}$ &  & $\lambda_{4}$ & $\lambda_{4}$\\
$\lambda_{3}$ &  & $\lambda_{4}$ &  & $\lambda_{4}$\\
$\lambda_{4}$ & $\lambda_{4}$ & $\lambda_{4}$ & $\lambda_{4}$ & $\lambda_{4}$%
\end{tabular}
\label{tabla_typeII}%
\end{equation}
produces, after $0_S$-reduction of the resonant subalgebra of the $S$-expanded
algebra of the $2$-dimensional isometry (\ref{2dimalg_dos}), an algebra which
coincides with $3$-dimensional Bianchi type II isometry.
\end{theorem}

\begin{proof}
To reproduce the type II algebra we have to choose, for example,
\begin{equation}
\lambda_{1}\lambda_{2}=\lambda_{3}\text{ and }\lambda_{2}\lambda_{3}%
=\lambda_{4} \label{sec_cond}%
\end{equation}
In that case the commutation relations (\ref{seg})  take the form%
\begin{align*}
\left[  \lambda_{2}X_{2},\lambda_{1}X_{1}\right]   &  =-\lambda_{3}X_{1}\\
\left[  \lambda_{2}X_{2},\lambda_{3}X_{1}\right]   &  =0\\
\left[  \lambda_{1}X_{1},\lambda_{3}X_{1}\right]   &  =0
\end{align*}
and renaming the generators as
\begin{align*}
Y_{1}  &  =\lambda_{3}X_{1}\\
Y_{2}  &  =\lambda_{1}X_{1}\\
Y_{3}  &  =\lambda_{2}X_{2}%
\end{align*}
we obtain the type II algebra%
\[
\left[  Y_{1},Y_{2}\right]  =\left[  Y_{1},Y_{3}\right]  =0,\ \ \ \left[
Y_{2},Y_{3}\right]  =Y_{1}%
\]

But in order for this result to be true, we must provide
 an explicit semigroup that satisfies the conditions
(\ref{sec_cond}). Until now our table has the form%
\[%
\begin{tabular}
[c]{l|llll}
& $\lambda_{1}$ & $\lambda_{2}$ & $\lambda_{3}$ & $\lambda_{4}$\\\hline
$\lambda_{1}$ &  & $\lambda_{3}$ &  & $\lambda_{4}$\\
$\lambda_{2}$ & $\lambda_{3}$ &  & $\lambda_{4}$ & $\lambda_{4}$\\
$\lambda_{3}$ &  & $\lambda_{4}$ &  & $\lambda_{4}$\\
$\lambda_{4}$ & $\lambda_{4}$ & $\lambda_{4}$ & $\lambda_{4}$ & $\lambda_{4}$%
\end{tabular}
\]
and the empty spaces must be filled in such a way that this table defines an
associative, commutative product and such that the decomposition
(\ref{newdes}) satisfies the resonant condition (\ref{res_con}).
\end{proof}

\begin{proposition}
There are semigroups that fit multiplication table \ref{tabla_typeII} and  resonant condition
(\ref{newdes}).
\end{proposition}

After looking for different possibilities we have found one way to fill the
table (\ref{tabla_typeII}). The proposed semigroup is%
\begin{equation}%
\begin{tabular}
[c]{l|llll}
& $\lambda_{1}$ & $\lambda_{2}$ & $\lambda_{3}$ & $\lambda_{4}$\\\hline
$\lambda_{1}$ & $\lambda_{2}$ & $\lambda_{3}$ & $\lambda_{4}$ & $\lambda_{4}%
$\\
$\lambda_{2}$ & $\lambda_{3}$ & $\lambda_{4}$ & $\lambda_{4}$ & $\lambda_{4}%
$\\
$\lambda_{3}$ & $\lambda_{4}$ & $\lambda_{4}$ & $\lambda_{4}$ & $\lambda_{4}%
$\\
$\lambda_{4}$ & $\lambda_{4}$ & $\lambda_{4}$ & $\lambda_{4}$ & $\lambda_{4}$%
\end{tabular}
\ \label{SN2}%
\end{equation}

That this multiplication table represents in fact an abelian semigroup can be checked
directly. That it is also commutative is seen  from the table. The associativity is proved by
a tedious but direct calculation.

Note that there may be other semigroups that can also lead to the type II
algebra. Those  correspond to other ways to fill the empty spaces
in table (\ref{tabla_typeII}).

\subsubsection{Type V and the $S_{N3}$ semigroup}

\begin{theorem}
An abelian semigroup of four elements $S=\left\{ \lambda_{1},\lambda_{2},\lambda_{3},\lambda_{4}\right\} $ 
satisfying the conditions of Lemma 5 and possessing the following constraints in the multiplication table
\begin{equation}%
\begin{tabular}
[c]{l|llll}
& $\lambda_{1}$ & $\lambda_{2}$ & $\lambda_{3}$ & $\lambda_{4}$\\\hline
$\lambda_{1}$ &  & $\lambda_{1}$ &  & $\lambda_{4}$\\
$\lambda_{2}$ & $\lambda_{1}$ &  & $\lambda_{3}$ & $\lambda_{4}$\\
$\lambda_{3}$ &  & $\lambda_{3}$ &  & $\lambda_{4}$\\
$\lambda_{4}$ & $\lambda_{4}$ & $\lambda_{4}$ & $\lambda_{4}$ & $\lambda_{4}$%
\end{tabular}
\label{tabla_typeV}%
\end{equation}
produces, after $0_S$-reduction of the resonant subalgebra of the $S$-expanded
algebra of the $2$-dimensional isometry (\ref{2dimalg_dos}), an algebra which
coincides with $3$-dimensional Bianchi type V isometry.
\end{theorem}

\begin{proof}
In fact, if we choose in the commutation relations (\ref{seg}), for example,
\begin{equation}
\lambda_{1}\lambda_{2}=\lambda_{1}\text{ and }\lambda_{2}\lambda_{3}%
=\lambda_{3} \label{five_cond}%
\end{equation}
In that case the commutation relations (\ref{seg})  take the form%
\begin{align*}
\left[  \lambda_{2}X_{2},\lambda_{1}X_{1}\right]   &  =-\lambda_{1}X_{1}\\
\left[  \lambda_{2}X_{2},\lambda_{3}X_{1}\right]   &  =-\lambda_{3}X_{1}\\
\left[  \lambda_{1}X_{1},\lambda_{3}X_{1}\right]   &  =0
\end{align*}
and renaming the generators as
\begin{align*}
Y_{1}  &  =\lambda_{1}X_{1}\\
Y_{2}  &  =\lambda_{3}X_{1}\\
Y_{3}  &  =\lambda_{2}X_{2}%
\end{align*}
we obtain the type V algebra%
\[
\left[  Y_{1},Y_{2}\right]  =0,\ \ \ \left[  Y_{1},Y_{3}\right]
=Y_{1},\ \ \ \left[  Y_{2},Y_{3}\right]  =Y_{2}%
\]

But again, in order for this result to be true we must provide of an
explicit semigroup that satisfies the conditions (\ref{five_cond}). Until now
our table has had the form%
\[%
\begin{tabular}
[c]{l|llll}
& $\lambda_{1}$ & $\lambda_{2}$ & $\lambda_{3}$ & $\lambda_{4}$\\\hline
$\lambda_{1}$ &  & $\lambda_{1}$ &  & $\lambda_{4}$\\
$\lambda_{2}$ & $\lambda_{1}$ &  & $\lambda_{3}$ & $\lambda_{4}$\\
$\lambda_{3}$ &  & $\lambda_{3}$ &  & $\lambda_{4}$\\
$\lambda_{4}$ & $\lambda_{4}$ & $\lambda_{4}$ & $\lambda_{4}$ & $\lambda_{4}$%
\end{tabular}
\]
and the empty spaces must be filled in a way that respects the required
conditions. Note that there are $4^{4}=256$ possibilities to fill this table in
a closed form. This number is reduced by imposing 
associativity, commutativity and the resonant condition for the decomposition
(\ref{newdes}). This number can be even reduced by the 
associativity condition.
\end{proof}

\begin{proposition}
There are semigroups that fit multiplication table \ref{tabla_typeV} and resonant condition 
(\ref{newdes}).
\end{proposition}

After studying different possibilities we found one way to fill the table
(\ref{tabla_typeV}). The proposed semigroup is%
\[%
\begin{tabular}
[c]{l|llll}
& $\lambda_{1}$ & $\lambda_{2}$ & $\lambda_{3}$ & $\lambda_{4}$\\\hline
$\lambda_{1}$ & $\lambda_{4}$ & $\lambda_{1}$ & $\lambda_{4}$ & $\lambda_{4}%
$\\
$\lambda_{2}$ & $\lambda_{1}$ & $\lambda_{2}$ & $\lambda_{3}$ & $\lambda_{4}%
$\\
$\lambda_{3}$ & $\lambda_{4}$ & $\lambda_{3}$ & $\lambda_{4}$ & $\lambda_{4}%
$\\
$\lambda_{4}$ & $\lambda_{4}$ & $\lambda_{4}$ & $\lambda_{4}$ & $\lambda_{4}$%
\end{tabular}
\
\]

That this multiplication represents an abelian semigroup can be
checked again by direct calculation.

We point out again that there may be other semigroups that can also lead to
the type V algebra. Those  correspond to other ways to fill the
empty spaces in table (\ref{tabla_typeV}).

\subsection{\textbf{The type I algebra}}

\label{typeI}

Starting from the abelian $2$-dimensional algebra%
\begin{equation}
\left[  X_{1},X_{2}\right]  =0 \label{I1}%
\end{equation}
we note that it also possesses the subspace structure (\ref{subs_structure})
where $V_{0}=\left\{  X_{2}\right\}  $ and $V_{1}=\left\{  X_{1}\right\}  $.
So, for example, by choosing the semigroup $S_{K}^{\left(  3\right)  }$ or
$S_{N1}$ both having a resonant decomposition of the form%
\begin{align}
S_{0}  &  =\left\{  \lambda_{2},\lambda_{3},\lambda_{4}\right\} \label{I2}\\
S_{1}  &  =\left\{  \lambda_{1},\lambda_{4}\right\} \nonumber
\end{align}
we obtain that the reduction of the resonant subalgebra
\begin{equation}
\mathcal{G}_{S,R}^{\text{red}}=\left\{  \lambda_{2}X_{2},\ \lambda_{3}%
X_{2},\ \lambda_{1}X_{1}\right\}  \label{I3}%
\end{equation}
will have the following commutation relations:%
\begin{align}
\left[  \lambda_{2}X_{2},\lambda_{3}X_{2}\right]   &  =\lambda_{2}\lambda
_{3}\left[  X_{2},X_{2}\right]  =0\label{I4}\\
\left[  \lambda_{2}X_{2},\lambda_{1}X_{1}\right]   &  =\lambda_{1}\lambda
_{2}\left[  X_{2},X_{1}\right]  =0\nonumber\\
\left[  \lambda_{3}X_{2},\lambda_{1}X_{1}\right]   &  =\lambda_{1}\lambda
_{3}\left[  X_{2},X_{1}\right]  =0\nonumber
\end{align}
what means that it doesn't matter if we use the semigroup $S_{K}^{\left(
3\right)  }$ or $S_{N1}$, the result will be always an abelian algebra in $3$
dimensions because the original algebra is abelian. The same result can be
reached with the semigroup $S_{E}^{\left(  2\right)  }$ whose semigroup
decomposition is similar to (\ref{I2}).

Also, by using the semigroups $S_{N2}$, $S_{N3}$ and probably others that have a
resonant decomposition of the form%
\begin{align}
S_{0}  &  =\left\{  \lambda_{2},\lambda_{4}\right\} \label{I5}\\
S_{1}  &  =\left\{  \lambda_{1},\lambda_{3},\lambda_{4}\right\} \nonumber
\end{align}
we obtain to a reduction of the resonant subalgebra
\[
\mathcal{G}_{S,R}^{\text{red}}=\left\{  \lambda_{2}X_{2},\ \lambda_{1}%
X_{1},\ \lambda_{3}X_{1}\right\}
\]
whose commutation relations%
\begin{align*}
\left[  \lambda_{2}X_{2},\lambda_{1}X_{1}\right]   &  =\lambda_{1}\lambda
_{2}\left[  X_{2},X_{1}\right]  =0\\
\left[  \lambda_{2}X_{2},\lambda_{3}X_{1}\right]   &  =\lambda_{2}\lambda
_{3}\left[  X_{2},X_{1}\right]  =0\\
\left[  \lambda_{1}X_{1},\lambda_{3}X_{1}\right]   &  =\lambda_{1}\lambda
_{3}\left[  X_{1},X_{1}\right]  =0
\end{align*}
are again no more than the $3$-dimensional abelian algebra.

So we conclude that starting from (\ref{I1}) whatever semigroup with a
zero element and that have a resonant decomposition of the form
(\ref{I2}) or (\ref{I5}) will lead to the type I algebra. Moreover, this
result can be generalized:

\begin{proposition}
An abelian algebra in $d$ dimensions can be obtained as an expansion of the
abelian algebra in $2$-dimensions by using a semigroup with probably a
 zero element and a suitable resonant decomposition.
\end{proposition}

Note that a crucial property to relate a $3$-dimensional algebra (whichever of
the type I, II, III and V) with a 2-dimensional algebra is the existence of the
resonant subalgebra and the $0_S$-reduction. This is the only way to obtain
three generators starting from two.

\subsection{\textbf{Brief summary}}

\label{Bsum}

Starting from
\begin{equation}
\left[  X_{1},X_{2}\right]  =0 \label{sum1}%
\end{equation}
it is possible to obtain the type I abelian algebra in three dimensions using
many semigroups as for example $S_{E}^{\left(  2\right)  }$, $S_{K}^{\left(
3\right)  }$, $S_{N1}$, $S_{N2}$, $S_{N3}$ and probably others. Now starting
from%
\begin{equation}
\left[  X_{1},X_{2}\right]  =X_{1} \label{sum2}%
\end{equation}
it is also possible to obtain the type I abelian algebra in three dimensions
using for example a semigroup whose multiplication satisfies the condition
$(i)$ of section \ref{BtypeIII}, i.e., whose table has the form
\[%
\begin{tabular}
[c]{l|llll}
& $\lambda_{1}$ & $\lambda_{2}$ & $\lambda_{3}$ & $\lambda_{4}$\\\hline
$\lambda_{1}$ &  & $\lambda_{4}$ & $\lambda_{4}$ & $\lambda_{4}$\\
$\lambda_{2}$ & $\lambda_{4}$ &  &  & $\lambda_{4}$\\
$\lambda_{3}$ & $\lambda_{4}$ &  &  & $\lambda_{4}$\\
$\lambda_{4}$ & $\lambda_{4}$ & $\lambda_{4}$ & $\lambda_{4}$ & $\lambda_{4}$%
\end{tabular}
\ \
\]
where the empty spaces must be filled with the corresponding conditions of
associativity, resonant condition, reduction condition, etc.

The semigroups with which it is possible to generate the type I, II, III and V algebra starting
from the $2$-dimensional algebra (\ref{sum2}) appear in the following table:%
\begin{equation}%
\begin{tabular}
[c]{|l|l|}\hline
\textbf{Algebra} & \textbf{Semigroup used}\\\hline
Type I &
\begin{tabular}
[c]{l}%
whatever semigroup with $0$-element\\
and a resonant decomposition
\end{tabular}
\\\hline
Type II & $S_{N2}$ and probably others\\\hline
Type III & $S_{E}^{\left(  2\right)  }$, $S_{K}^{\left(  3\right)  }$,
$S_{N1}$ and probably others\\\hline
Type V & $S_{N3}$ and probably others\\\hline
\end{tabular}
\ \ \ \ \ \ \ \label{Tab2}%
\end{equation}
where the mentioned semigroups are described in the following table%
\begin{equation}%
\begin{tabular}
[c]{|l|l|l|l|}\hline
Semigroup & Table of multiplicarion & Resonant decomposition & $0_{S}%
$-element\\\hline
\multicolumn{1}{|c|}{$S_{E}^{\left(  2\right)  }$} & \multicolumn{1}{|c|}{$%
\begin{tabular}
[c]{l|llll}
& $\lambda_{0}$ & $\lambda_{1}$ & $\lambda_{2}$ & $\lambda_{3}$\\\hline
$\lambda_{0}$ & $\lambda_{0}$ & $\lambda_{1}$ & $\lambda_{2}$ & $\lambda_{3}%
$\\
$\lambda_{1}$ & $\lambda_{1}$ & $\lambda_{2}$ & $\lambda_{3}$ & $\lambda_{3}%
$\\
$\lambda_{2}$ & $\lambda_{2}$ & $\lambda_{3}$ & $\lambda_{3}$ & $\lambda_{3}%
$\\
$\lambda_{3}$ & $\lambda_{3}$ & $\lambda_{3}$ & $\lambda_{3}$ & $\lambda_{3}$%
\end{tabular}
\ \ \ \ \ \ \ $} & \multicolumn{1}{|c|}{%
\begin{tabular}
[c]{l}%
$S_{0}=\left\{  \lambda_{0},\lambda_{2},\lambda_{3}\right\}  $\\
$S_{1}=\left\{  \lambda_{1},\lambda_{3}\right\}  $%
\end{tabular}
} & \multicolumn{1}{|c|}{$\lambda_{3}$}\\\hline
\multicolumn{1}{|c|}{$S_{K}^{\left(  3\right)  }$} & \multicolumn{1}{|c|}{$%
\begin{tabular}
[c]{l|llll}
& $\lambda_{1}$ & $\lambda_{2}$ & $\lambda_{3}$ & $\lambda_{4}$\\\hline
$\lambda_{1}$ & $\lambda_{4}$ & $\lambda_{4}$ & $\lambda_{1}$ & $\lambda_{4}%
$\\
$\lambda_{2}$ & $\lambda_{4}$ & $\lambda_{2}$ & $\lambda_{2}$ & $\lambda_{4}%
$\\
$\lambda_{3}$ & $\lambda_{1}$ & $\lambda_{2}$ & $\lambda_{3}$ & $\lambda_{4}%
$\\
$\lambda_{4}$ & $\lambda_{4}$ & $\lambda_{4}$ & $\lambda_{4}$ & $\lambda_{4}$%
\end{tabular}
\ \ \ \ \ \ \ $} & \multicolumn{1}{|c|}{%
\begin{tabular}
[c]{l}%
$S_{0}=\left\{  \lambda_{2},\lambda_{3},\lambda_{4}\right\}  $\\
$S_{1}=\left\{  \lambda_{1},\lambda_{4}\right\}  $%
\end{tabular}
} & \multicolumn{1}{|c|}{$\lambda_{4}$}\\\hline
\multicolumn{1}{|c|}{$S_{N1}$} & \multicolumn{1}{|c|}{$%
\begin{tabular}
[c]{l|llll}
& $\lambda_{1}$ & $\lambda_{2}$ & $\lambda_{3}$ & $\lambda_{4}$\\\hline
$\lambda_{1}$ & $\lambda_{4}$ & $\lambda_{4}$ & $\lambda_{1}$ & $\lambda_{4}%
$\\
$\lambda_{2}$ & $\lambda_{4}$ & $\lambda_{2}$ & $\lambda_{4}$ & $\lambda_{4}%
$\\
$\lambda_{3}$ & $\lambda_{1}$ & $\lambda_{4}$ & $\lambda_{3}$ & $\lambda_{4}%
$\\
$\lambda_{4}$ & $\lambda_{4}$ & $\lambda_{4}$ & $\lambda_{4}$ & $\lambda_{4}$%
\end{tabular}
\ \ \ \ \ $ $\ $} & \multicolumn{1}{|c|}{%
\begin{tabular}
[c]{l}%
$S_{0}=\left\{  \lambda_{2},\lambda_{3},\lambda_{4}\right\}  $\\
$S_{1}=\left\{  \lambda_{1},\lambda_{4}\right\}  $%
\end{tabular}
} & \multicolumn{1}{|c|}{$\lambda_{4}$}\\\hline
\multicolumn{1}{|c|}{$S_{N2}$} & \multicolumn{1}{|c|}{$%
\begin{tabular}
[c]{l|llll}
& $\lambda_{1}$ & $\lambda_{2}$ & $\lambda_{3}$ & $\lambda_{4}$\\\hline
$\lambda_{1}$ & $\lambda_{2}$ & $\lambda_{3}$ & $\lambda_{4}$ & $\lambda_{4}%
$\\
$\lambda_{2}$ & $\lambda_{3}$ & $\lambda_{4}$ & $\lambda_{4}$ & $\lambda_{4}%
$\\
$\lambda_{3}$ & $\lambda_{4}$ & $\lambda_{4}$ & $\lambda_{4}$ & $\lambda_{4}%
$\\
$\lambda_{4}$ & $\lambda_{4}$ & $\lambda_{4}$ & $\lambda_{4}$ & $\lambda_{4}$%
\end{tabular}
\ \ \ \ \ \ \ $} & \multicolumn{1}{|c|}{%
\begin{tabular}
[c]{l}%
$S_{0}=\left\{  \lambda_{2},\lambda_{4}\right\}  $\\
$S_{1}=\left\{  \lambda_{1},\lambda_{3},\lambda_{4}\right\}  $%
\end{tabular}
} & \multicolumn{1}{|c|}{$\lambda_{4}$}\\\hline
\multicolumn{1}{|c|}{$S_{N3}$} & \multicolumn{1}{|c|}{$%
\begin{tabular}
[c]{l|llll}
& $\lambda_{1}$ & $\lambda_{2}$ & $\lambda_{3}$ & $\lambda_{4}$\\\hline
$\lambda_{1}$ & $\lambda_{4}$ & $\lambda_{1}$ & $\lambda_{4}$ & $\lambda_{4}%
$\\
$\lambda_{2}$ & $\lambda_{1}$ & $\lambda_{2}$ & $\lambda_{3}$ & $\lambda_{4}%
$\\
$\lambda_{3}$ & $\lambda_{4}$ & $\lambda_{3}$ & $\lambda_{4}$ & $\lambda_{4}%
$\\
$\lambda_{4}$ & $\lambda_{4}$ & $\lambda_{4}$ & $\lambda_{4}$ & $\lambda_{4}$%
\end{tabular}
\ \ \ \ \ \ \ $} & \multicolumn{1}{|c|}{%
\begin{tabular}
[c]{l}%
$S_{0}=\left\{  \lambda_{2},\lambda_{4}\right\}  $\\
$S_{1}=\left\{  \lambda_{1},\lambda_{3},\lambda_{4}\right\}  $%
\end{tabular}
} & \multicolumn{1}{|c|}{$\lambda_{4}$}\\\hline
\end{tabular}
\ \ \ \ \ \ \ \label{Tab3}%
\end{equation}

\section{\textbf{The Bianchi spaces not-related with $2$-dimensional
isometries}}

\label{Bianchi_no_related}

\subsection{\textbf{Type IV, VI, VII$_{2}$, VIII and IX algebras}}

\label{typeIV}

Let's consider for example the type IV algebra%
\begin{align}
\left[  Y_{1},Y_{2}\right]   &  =0,\label{not1}\\
\left[  Y_{1},Y_{3}\right]   &  =Y_{1},\label{not2}\\
\left[  Y_{2},Y_{3}\right]   &  =Y_{1}+Y_{2}. \label{not3}%
\end{align}
As the $S$-expansion method uses an induced bracket
\[
\left[  \lambda_{\alpha}X_{i},\lambda_{\beta}X_{j}\right]  =\lambda_{\alpha
}\lambda_{\beta}\left[  X_{i},X_{j}\right]  =\lambda_{\gamma\left(
\alpha,\beta\right)  }\left[  X_{i},X_{j}\right]
\]
for the expanded algebra and considering that for our original algebra
$i,j=1,2$ and its commutation relation is given by%
\[
\left[  X_{1},X_{2}\right]  =X_{1},%
\]
we have that the first two relations (\ref{not1},\ref{not2}) can easily be
reproduced with some semigroup product, but to reproduce (\ref{not3}) we need
a non-zero result as the first requirement. This means
that we must have a relation like%
\[
\left[  \lambda_{\alpha}X_{1},\lambda_{\beta}X_{2}\right]  =\lambda_{\alpha
}\lambda_{\beta}\left[  X_{1},X_{2}\right]  =\lambda_{\gamma\left(
\alpha,\beta\right)  }X_{1}.%
\]
And here we can see that no matter which semigroup we choose, $\lambda
_{\gamma\left(  \alpha,\beta\right)  }$ will always be an element of the
semigroup (it is closed) and therefore  we will never be able to
reproduce a sum of two generators.

Now consider the type VI algebra.
\begin{align*}
\left[  Y_{1},Y_{2}\right]   &  =0,\\
\left[  Y_{1},Y_{3}\right]   &  =Y_{1},\\
\left[  Y_{2},Y_{3}\right]   &  =hY_{2},\ \ \ h\neq0,1.
\end{align*}
Again the first two brackets could be reproduced by a certain semigroup, but
for the third one we would have something like%
\[
\left[  \lambda_{\alpha}X_{1},\lambda_{\beta}X_{2}\right]  =\lambda
_{\gamma\left(  \alpha,\beta\right)  }X_{1}%
\]
and again, no matter which semigroup we choose, $\lambda_{\gamma\left(
\alpha,\beta\right)  }$ will always be an element of the semigroup and
 we will never be able to reproduce semigroup element multiplied
by a numeric factor.  A similar argument can be used to show that  type
VII$_{1}$ algebra cannot be obtained by the $S$-expansion procedure.

A mix of the above arguments can be applied to explain why it is impossible to
obtain the type VII$_{2}$ and VIII algebras as an expansion of a
$2$-dimensional isometry.

Finally, to show why it is also impossible to reproduce the type IX algebra
\[
\left[  Y_{1},Y_{2}\right]  =Y_{3},\ \ \ \left[  Y_{2},Y_{3}\right]
=Y_{1},\ \ \ \left[  Y_{3},Y_{1}\right]  =Y_{2}%
\]
we have to realize that the candidate for being the expanded algebra will have three
commutation relations of the form%
\[
\left[  \lambda_{\alpha}X_{i},\lambda_{\beta}X_{j}\right]  =\lambda
_{\gamma\left(  \alpha,\beta\right)  }\left[  X_{i},X_{j}\right]
\]
but where $i,j$ takes the values $1$ and $2$. Therefore, in one of the three commutation relations one
index will always be repeated leading to a vanishing bracket. So it is impossible to generate, 
by means of an $S$-expansion, a 3-dimensional algebra with the three brackets having a non-
zero value.

Thus, we conclude that these types of algebra, that cannot be obtained by an
expansion of the 2-dimensional isometries, are in some sense intrinsic in $3$ dimensions.

\section{\textbf{Checking with computer programs}}

\label{check}

A common question when working with semigroups in section
\ref{Bianchi_related} is that of the existence of diverse semigroups given a
multiplication table with some elements already chosen. We have, for example,
a table like this one%
\[%
\begin{tabular}
[c]{l|llll}%
$S^{?}$ & $\lambda_{1}$ & $\lambda_{2}$ & $\lambda_{3}$ & $\lambda_{4}%
$\\\hline
$\lambda_{1}$ &  & $\lambda_{4}$ & $\lambda_{1}$ & $\lambda_{4}$\\
$\lambda_{2}$ & $\lambda_{4}$ &  &  & $\lambda_{4}$\\
$\lambda_{3}$ & $\lambda_{1}$ &  &  & $\lambda_{4}$\\
$\lambda_{4}$ & $\lambda_{4}$ & $\lambda_{4}$ & $\lambda_{4}$ & $\lambda_{4}$%
\end{tabular}
\]

In principle there are 256 different symmetric matrices which fill this
template, but not all of them will be semigroups because the multiplication
table will not always be associative. Moreover, we have to select only those
that satisfy a certain \textit{resonant condition}. Finally, many of these
associative tables will be isomorphic, so we only have to select those that
are not.

In what follows we find all the non isomorphic forms to fill the tables
(\ref{case3}), (\ref{tabla_typeII}) and (\ref{tabla_typeV}) with the mentioned
conditions and show that all the semigroups given in table (\ref{Tab2}) (those
semigroups that we have constructed by hand) are isomorphic to one of the
semigroups given by the computer program \textit{com.f} of \cite{Plemmons}.

\subsection{\textbf{Type II}}

The template is:
\[%
\begin{tabular}
[c]{l|llll}%
$S^{?}$ & $\lambda_{1}$ & $\lambda_{2}$ & $\lambda_{3}$ & $\lambda_{4}%
$\\\hline
$\lambda_{1}$ &  & $\lambda_{3}$ &  & $\lambda_{4}$\\
$\lambda_{2}$ & $\lambda_{3}$ &  & $\lambda_{4}$ & $\lambda_{4}$\\
$\lambda_{3}$ &  & $\lambda_{4}$ &  & $\lambda_{4}$\\
$\lambda_{4}$ & $\lambda_{4}$ & $\lambda_{4}$ & $\lambda_{4}$ & $\lambda_{4}$%
\end{tabular}
\]

By using computer programs, we have found that there are two non isomorphic
ways of filling this template such that: $a)$ the resulting table is an
abelian semigroup and $b)$ the resonant decomposition is given by
\begin{equation}
S_{0}=\{\lambda_{2},\lambda_{4}\},\ S_{1}=\{\lambda_{1},\lambda_{3}%
,\lambda_{4}\}\text{.}\label{ch1}%
\end{equation}
Those ways are:%
\begin{equation}%
\begin{tabular}
[c]{l|llll}%
$S_{II}^{1}$ & $\lambda_{1}$ & $\lambda_{2}$ & $\lambda_{3}$ & $\lambda_{4}%
$\\\hline
$\lambda_{1}$ & $\lambda_{4}$ & $\lambda_{3}$ & $\lambda_{4}$ & $\lambda_{4}%
$\\
$\lambda_{2}$ & $\lambda_{3}$ & $\lambda_{4}$ & $\lambda_{4}$ & $\lambda_{4}%
$\\
$\lambda_{3}$ & $\lambda_{4}$ & $\lambda_{4}$ & $\lambda_{4}$ & $\lambda_{4}%
$\\
$\lambda_{4}$ & $\lambda_{4}$ & $\lambda_{4}$ & $\lambda_{4}$ & $\lambda_{4}$%
\end{tabular}
\ \ \ \ \ \ ;\text{\ \ }\
\begin{tabular}
[c]{l|llll}%
$S_{II}^{2}$ & $\lambda_{1}$ & $\lambda_{2}$ & $\lambda_{3}$ & $\lambda_{4}%
$\\\hline
$\lambda_{1}$ & $\lambda_{2}$ & $\lambda_{3}$ & $\lambda_{4}$ & $\lambda_{4}%
$\\
$\lambda_{2}$ & $\lambda_{3}$ & $\lambda_{4}$ & $\lambda_{4}$ & $\lambda_{4}%
$\\
$\lambda_{3}$ & $\lambda_{4}$ & $\lambda_{4}$ & $\lambda_{4}$ & $\lambda_{4}%
$\\
$\lambda_{4}$ & $\lambda_{4}$ & $\lambda_{4}$ & $\lambda_{4}$ & $\lambda_{4}$%
\end{tabular}
\ \ \ \ \text{.}\label{ch_1__1}%
\end{equation}
Each of them is isomorphic to one of the semigroups of the list given by the
program \textit{com.f} of \cite{Plemmons}\ for $n=4$. We give this information
in the following table,%

\begin{equation}%
\begin{tabular}
[c]{lll|c}
& isomorphic to &  & isomorphism\\\hline
$S_{II}^{1}$ & \multicolumn{1}{c}{$\Longleftrightarrow$} & $S_{\left(
4\right)  }^{10}$ & $(\ \lambda_{4}\ \lambda_{3}\ \lambda_{1}\ \lambda_{2}%
\ )$\\
$S_{II}^{2}$ & \multicolumn{1}{c}{$\Longleftrightarrow$} & $S_{\left(
4\right)  }^{12}$ & $(\ \lambda_{4}\ \lambda_{3}\ \lambda_{2}\ \lambda_{1}\ )$%
\end{tabular}
\label{ch_1_2}%
\end{equation}
where the isomorphism denoted by $(\ \lambda_{a}\ \lambda_{b}\ \lambda
_{c}\ \lambda_{d}\ )$ means: change $\lambda_{1}$ by $\lambda_{a}$,
$\lambda_{2}$ by $\lambda_{b}$, $\lambda_{3}$ by $\lambda_{c}$ and
$\lambda_{4}$ by $\lambda_{d}$. The semigroups $S_{\left(  4\right)  }^{10}$
and $S_{\left(  4\right)  }^{12}$ of the list given by the program
\textit{com.f} for $n=4$ are:%
\begin{equation}%
\begin{tabular}
[c]{l|llll}%
$S_{\left(  4\right)  }^{10}$ & $\lambda_{1}$ & $\lambda_{2}$ & $\lambda_{3}$
& $\lambda_{4}$\\\hline
$\lambda_{1}$ & $\lambda_{1}$ & $\lambda_{1}$ & $\lambda_{1}$ & $\lambda_{1}%
$\\
$\lambda_{2}$ & $\lambda_{1}$ & $\lambda_{1}$ & $\lambda_{1}$ & $\lambda_{1}%
$\\
$\lambda_{3}$ & $\lambda_{1}$ & $\lambda_{1}$ & $\lambda_{1}$ & $\lambda_{2}%
$\\
$\lambda_{4}$ & $\lambda_{1}$ & $\lambda_{1}$ & $\lambda_{2}$ & $\lambda_{1}$%
\end{tabular}
\ \ \ \text{\ ; \ }%
\begin{tabular}
[c]{l|llll}%
$S_{\left(  4\right)  }^{12}$ & $\lambda_{1}$ & $\lambda_{2}$ & $\lambda_{3}$
& $\lambda_{4}$\\\hline
$\lambda_{1}$ & $\lambda_{1}$ & $\lambda_{1}$ & $\lambda_{1}$ & $\lambda_{1}%
$\\
$\lambda_{2}$ & $\lambda_{1}$ & $\lambda_{1}$ & $\lambda_{1}$ & $\lambda_{1}%
$\\
$\lambda_{3}$ & $\lambda_{1}$ & $\lambda_{1}$ & $\lambda_{1}$ & $\lambda_{2}%
$\\
$\lambda_{4}$ & $\lambda_{1}$ & $\lambda_{1}$ & $\lambda_{2}$ & $\lambda_{3}$%
\end{tabular}
\ \ \ \text{.}\label{ch_1_3}%
\end{equation}

It can be checked directly that applying the isomorphism $(\ \lambda
_{4}\ \lambda_{3}\ \lambda_{1}\ \lambda_{2}\ )$ to $S^{10}$ obtains
$S_{II}^{1}$ and applying the isomorphism $(\ \lambda_{4}\ \lambda
_{3}\ \lambda_{2}\ \lambda_{1}\ )$ to $S^{12}$ obtains $S_{II}^{2}$.

\subsection{\textbf{Type III}}

The template is:
\[%
\begin{tabular}
[c]{l|llll}%
$S^{?}$ & $\lambda_{1}$ & $\lambda_{2}$ & $\lambda_{3}$ & $\lambda_{4}%
$\\\hline
$\lambda_{1}$ &  & $\lambda_{4}$ & $\lambda_{1}$ & $\lambda_{4}$\\
$\lambda_{2}$ & $\lambda_{4}$ &  &  & $\lambda_{4}$\\
$\lambda_{3}$ & $\lambda_{1}$ &  &  & $\lambda_{4}$\\
$\lambda_{4}$ & $\lambda_{4}$ & $\lambda_{4}$ & $\lambda_{4}$ & $\lambda_{4}$%
\end{tabular}
\ \
\]

We have found that there are $7$ non-isomorphic ways of filling this template
such that: $a)$ the resulting table is an abelian semigroup and $b)$ the
resonant decomposition is given by %

\begin{equation}
S_{0}=\{\lambda_{2},\lambda_{3},\lambda_{4}\},S_{1}=\{\lambda_{1},\lambda
_{4}\}\text{.}\label{ch2}%
\end{equation}
Those ways are:%

\[%
\begin{tabular}
[c]{l|llll}%
$S_{III}^{1}$ & $\lambda_{1}$ & $\lambda_{2}$ & $\lambda_{3}$ & $\lambda_{4}%
$\\\hline
$\lambda_{1}$ & $\lambda_{4}$ & $\lambda_{4}$ & $\lambda_{1}$ & $\lambda_{4}%
$\\
$\lambda_{2}$ & $\lambda_{4}$ & $\lambda_{4}$ & $\lambda_{4}$ & $\lambda_{4}%
$\\
$\lambda_{3}$ & $\lambda_{1}$ & $\lambda_{4}$ & $\lambda_{3}$ & $\lambda_{4}%
$\\
$\lambda_{4}$ & $\lambda_{4}$ & $\lambda_{4}$ & $\lambda_{4}$ & $\lambda_{4}$%
\end{tabular}
\ \ \ \text{\ ;\ }\
\begin{tabular}
[c]{l|llll}%
$S_{III}^{2}$ & $\lambda_{1}$ & $\lambda_{2}$ & $\lambda_{3}$ & $\lambda_{4}%
$\\\hline
$\lambda_{1}$ & $\lambda_{3}$ & $\lambda_{4}$ & $\lambda_{1}$ & $\lambda_{4}%
$\\
$\lambda_{2}$ & $\lambda_{4}$ & $\lambda_{4}$ & $\lambda_{4}$ & $\lambda_{4}%
$\\
$\lambda_{3}$ & $\lambda_{1}$ & $\lambda_{4}$ & $\lambda_{3}$ & $\lambda_{4}%
$\\
$\lambda_{4}$ & $\lambda_{4}$ & $\lambda_{4}$ & $\lambda_{4}$ & $\lambda_{4}$%
\end{tabular}
\ \ \ \text{\ ;\ }\
\begin{tabular}
[c]{l|llll}%
$S_{III}^{3}$ & $\lambda_{1}$ & $\lambda_{2}$ & $\lambda_{3}$ & $\lambda_{4}%
$\\\hline
$\lambda_{1}$ & $\lambda_{4}$ & $\lambda_{4}$ & $\lambda_{1}$ & $\lambda_{4}%
$\\
$\lambda_{2}$ & $\lambda_{4}$ & $\lambda_{4}$ & $\lambda_{2}$ & $\lambda_{4}%
$\\
$\lambda_{3}$ & $\lambda_{1}$ & $\lambda_{2}$ & $\lambda_{3}$ & $\lambda_{4}%
$\\
$\lambda_{4}$ & $\lambda_{4}$ & $\lambda_{4}$ & $\lambda_{4}$ & $\lambda_{4}$%
\end{tabular}
\
\]

\[%
\begin{tabular}
[c]{l|llll}%
$S_{III}^{4}$ & $\lambda_{1}$ & $\lambda_{2}$ & $\lambda_{3}$ & $\lambda_{4}%
$\\\hline
$\lambda_{1}$ & $\lambda_{2}$ & $\lambda_{4}$ & $\lambda_{1}$ & $\lambda_{4}%
$\\
$\lambda_{2}$ & $\lambda_{4}$ & $\lambda_{4}$ & $\lambda_{2}$ & $\lambda_{4}%
$\\
$\lambda_{3}$ & $\lambda_{1}$ & $\lambda_{2}$ & $\lambda_{3}$ & $\lambda_{4}%
$\\
$\lambda_{4}$ & $\lambda_{4}$ & $\lambda_{4}$ & $\lambda_{4}$ & $\lambda_{4}$%
\end{tabular}
\ \ \text{\ ;\ }\
\begin{tabular}
[c]{l|llll}%
$S_{III}^{5}$ & $\lambda_{1}$ & $\lambda_{2}$ & $\lambda_{3}$ & $\lambda_{4}%
$\\\hline
$\lambda_{1}$ & $\lambda_{4}$ & $\lambda_{4}$ & $\lambda_{1}$ & $\lambda_{4}%
$\\
$\lambda_{2}$ & $\lambda_{4}$ & $\lambda_{2}$ & $\lambda_{4}$ & $\lambda_{4}%
$\\
$\lambda_{3}$ & $\lambda_{1}$ & $\lambda_{4}$ & $\lambda_{3}$ & $\lambda_{4}%
$\\
$\lambda_{4}$ & $\lambda_{4}$ & $\lambda_{4}$ & $\lambda_{4}$ & $\lambda_{4}$%
\end{tabular}
\ \ \text{\ ;\ }\
\begin{tabular}
[c]{l|llll}%
$S_{III}^{6}$ & $\lambda_{1}$ & $\lambda_{2}$ & $\lambda_{3}$ & $\lambda_{4}%
$\\\hline
$\lambda_{1}$ & $\lambda_{4}$ & $\lambda_{4}$ & $\lambda_{1}$ & $\lambda_{4}%
$\\
$\lambda_{2}$ & $\lambda_{4}$ & $\lambda_{2}$ & $\lambda_{2}$ & $\lambda_{4}%
$\\
$\lambda_{3}$ & $\lambda_{1}$ & $\lambda_{2}$ & $\lambda_{3}$ & $\lambda_{4}%
$\\
$\lambda_{4}$ & $\lambda_{4}$ & $\lambda_{4}$ & $\lambda_{4}$ & $\lambda_{4}$%
\end{tabular}
\]

\[%
\begin{tabular}
[c]{l|llll}%
$S_{III}^{7}$ & $\lambda_{1}$ & $\lambda_{2}$ & $\lambda_{3}$ & $\lambda_{4}%
$\\\hline
$\lambda_{1}$ & $\lambda_{4}$ & $\lambda_{4}$ & $\lambda_{1}$ & $\lambda_{4}%
$\\
$\lambda_{2}$ & $\lambda_{4}$ & $\lambda_{2}$ & $\lambda_{4}$ & $\lambda_{4}%
$\\
$\lambda_{3}$ & $\lambda_{1}$ & $\lambda_{4}$ & $\lambda_{3}$ & $\lambda_{4}%
$\\
$\lambda_{4}$ & $\lambda_{4}$ & $\lambda_{4}$ & $\lambda_{4}$ & $\lambda_{4}$%
\end{tabular}
\
\]

As before, each of these forms are isomorphic to one of the semigroups of the
list given by the program \textit{com.f} of \cite{Plemmons}\ for $n=4$. Those
semigroups and the corresponding isomorphisms are given in the following table:%

\begin{equation}%
\begin{tabular}
[c]{lll|c}
& isomorphic to &  & isomorphism\\\hline
\multicolumn{1}{c}{$S_{III}^{1}$} & \multicolumn{1}{c}{$\Longleftrightarrow$}
& \multicolumn{1}{c|}{$S_{\left(  4\right)  }^{13}$} & $(\ \lambda
_{4}\ \lambda_{2}\ \lambda_{1}\ \lambda_{3}\ )$\\
\multicolumn{1}{c}{$S_{III}^{2}$} & \multicolumn{1}{c}{$\Longleftrightarrow$}
& \multicolumn{1}{c|}{$S_{\left(  4\right)  }^{28}$} & $(\ \lambda
_{4}\ \lambda_{2}\ \lambda_{3}\ \lambda_{1}\ )$\\
\multicolumn{1}{c}{$S_{III}^{3}$} & \multicolumn{1}{c}{$\Longleftrightarrow$}
& \multicolumn{1}{c|}{$S_{\left(  4\right)  }^{42}$} & $(\ \lambda
_{4}\ \lambda_{1}\ \lambda_{2}\ \lambda_{3}\ )$\\
\multicolumn{1}{c}{$S_{III}^{4}$} & \multicolumn{1}{c}{$\Longleftrightarrow$}
& \multicolumn{1}{c|}{$S_{\left(  4\right)  }^{43}$} & $(\ \lambda
_{4}\ \lambda_{2}\ \lambda_{1}\ \lambda_{3}\ )$\\
\multicolumn{1}{c}{$S_{III}^{5}$} & \multicolumn{1}{c}{$\Longleftrightarrow$}
& \multicolumn{1}{c|}{$S_{\left(  4\right)  }^{44}$} & $(\ \lambda
_{4}\ \lambda_{1}\ \lambda_{2}\ \lambda_{3}\ )$\\
\multicolumn{1}{c}{$S_{III}^{6}$} & \multicolumn{1}{c}{$\Longleftrightarrow$}
& \multicolumn{1}{c|}{$S_{\left(  4\right)  }^{45}$} & $(\ \lambda
_{4}\ \lambda_{1}\ \lambda_{2}\ \lambda_{3}\ )$\\
\multicolumn{1}{c}{$S_{III}^{7}$} & \multicolumn{1}{c}{$\Longleftrightarrow$}
& \multicolumn{1}{c|}{$S_{\left(  4\right)  }^{64}$} & $(\ \lambda
_{4}\ \lambda_{2}\ \lambda_{3}\ \lambda_{1}\ )$%
\end{tabular}
\ \ \ \ \ \ \label{iso2}%
\end{equation}
where the semigroups $S_{\left(  4\right)  }^{13}$, $S_{\left(  4\right)
}^{28}$, $S_{\left(  4\right)  }^{42}$, $S_{\left(  4\right)  }^{43}$,
$S_{\left(  4\right)  }^{44}$, $S_{\left(  4\right)  }^{45}$ y $S_{\left(
4\right)  }^{64}$, of the list generated by the program \textit{com.f} for
$n=4$, are explicitly given in the Appendix.

\subsection{\textbf{Type V}}

The template is:
\[%
\begin{tabular}
[c]{l|llll}%
$S^{?}$ & $\lambda_{1}$ & $\lambda_{2}$ & $\lambda_{3}$ & $\lambda_{4}%
$\\\hline
$\lambda_{1}$ &  & $\lambda_{1}$ &  & $\lambda_{4}$\\
$\lambda_{2}$ & $\lambda_{1}$ &  & $\lambda_{3}$ & $\lambda_{4}$\\
$\lambda_{3}$ &  & $\lambda_{3}$ &  & $\lambda_{4}$\\
$\lambda_{4}$ & $\lambda_{4}$ & $\lambda_{4}$ & $\lambda_{4}$ & $\lambda_{4}$%
\end{tabular}
\ \ \ \
\]
In this case we have found that there is just one way of filling this template
such that: $a)$ the resulting table is an abelian semigroup and $b)$ the
resonant decomposition is given by%
\begin{equation}
S_{0}=\{\lambda_{2},\lambda_{4}\},S_{1}=\{\lambda_{1},\lambda_{3},\lambda
_{4}\}\text{.}\label{ch3}%
\end{equation}
This is:%
\[%
\begin{tabular}
[c]{l|llll}%
$S_{V}$ & $\lambda_{1}$ & $\lambda_{2}$ & $\lambda_{3}$ & $\lambda_{4}%
$\\\hline
$\lambda_{1}$ & $\lambda_{4}$ & $\lambda_{1}$ & $\lambda_{4}$ & $\lambda_{4}%
$\\
$\lambda_{2}$ & $\lambda_{1}$ & $\lambda_{2}$ & $\lambda_{3}$ & $\lambda_{4}%
$\\
$\lambda_{3}$ & $\lambda_{4}$ & $\lambda_{3}$ & $\lambda_{4}$ & $\lambda_{4}%
$\\
$\lambda_{4}$ & $\lambda_{4}$ & $\lambda_{4}$ & $\lambda_{4}$ & $\lambda_{4}$%
\end{tabular}
\]
This table is isomorphic to the semigroup $S_{\left(  4\right)  }^{42}$ given
in Appendix. The isomorphism is given by
\begin{equation}
(\lambda_{4}\ \lambda_{1}\ \lambda_{3}\ \lambda_{2}\ ).\label{iso3}%
\end{equation}

Note that the semigroup $S_{\left(  4\right)  }^{42}$ also permits us to
obtain type III algebra. So we can ask, \textit{how can the same semigroup
lead at the same time to type III and V algebras?} The reason is that this
semigroup has two different resonant decompositions, (\ref{ch1}) and
(\ref{ch2}). Each of them permits us to extract different kinds of resonant
subalgebras leading, after the reduction, to completely different algebras.

\subsection{\textbf{Isomorphisms and consistency of the procedure}}

In the following table we summarize our results by specifying \textit{all} the
non-isomorphic semigroups that permit us to generate the type I, II, III and
the V algebra starting from the $2$-dimensional algebra (\ref{sum2}):%

\begin{equation}%
\begin{tabular}
[c]{|l|l|}\hline
\textbf{Algebra} & \textbf{Semigroup used}\\\hline
Type I & many semigroups (see section \ref{typeI})\\\hline
Type II & $S_{\left(  4\right)  }^{10}$, $S_{\left(  4\right)  }^{12}$\\\hline
Type III & $S_{\left(  4\right)  }^{13}$, $S_{\left(  4\right)  }^{28}$,
$S_{\left(  4\right)  }^{42}$, $S_{\left(  4\right)  }^{43}$, $S_{\left(
4\right)  }^{44}$, $S_{\left(  4\right)  }^{45}$ and $S_{\left(  4\right)
}^{64}$\\\hline
Type V & $S_{\left(  4\right)  }^{42}$\\\hline
\end{tabular}
\label{Final}%
\end{equation}

For consistency we should prove that each semigroup of the table (\ref{Tab2})
(which we have constructed by hand in section \ref{Bianchi_related}) is
isomorphic to one of the semigroups of table (\ref{Final}) that we have found
by using computer programs. This information is given in the following table:%

\[%
\begin{tabular}
[c]{lll|c}
& isomorphic to &  & isomorphism\\\hline
$S_{N1}$ & \multicolumn{1}{c}{$\Longleftrightarrow$} & $S_{\left(  4\right)
}^{44}$ & $(\ \lambda_{4}\ \lambda_{1}\ \lambda_{2}\ \lambda_{3}\ )$\\
$S_{N2}$ & \multicolumn{1}{c}{$\Longleftrightarrow$} & $S_{\left(  4\right)
}^{12}$ & $(\ \lambda_{4}\ \lambda_{3}\ \lambda_{2}\ \lambda_{1}\ )$\\
$S_{N3}$ & \multicolumn{1}{c}{$\Longleftrightarrow$} & $S_{\left(  4\right)
}^{42}$ & $(\ \lambda_{4}\ \lambda_{1}\ \lambda_{3}\ \lambda_{2}\ )$\\
$S_{E}^{\left(  2\right)  }$ & \multicolumn{1}{c}{$\Longleftrightarrow$} &
$S_{\left(  4\right)  }^{43}$ & $(\ \lambda_{4}\ \lambda_{3}\ \lambda
_{2}\ \lambda_{1}\ )$\\
$S_{K}^{\left(  3\right)  }$ & \multicolumn{1}{c}{$\Longleftrightarrow$} &
$S_{\left(  4\right)  }^{45}$ & $(\ \lambda_{4}\ \lambda_{2}\ \lambda
_{1}\ \lambda_{3}\ )$%
\end{tabular}
\]

\section{\textbf{Comments}}

\label{Com}

So, in this work  we present  a complete study about the
possibility of relate, by means of an expansion, the isometry algebras that
 act transitively in $2$ and $3$ dimensions. It was found that
some isometries in $3$ dimensions, specifically those of type I, II, III and V
(according  Bianchi's classification), can be obtained as
expansions of the isometries in $2$ dimensions. In general, there is more than
one possibility to obtain these results, i.e., it can happen that different
semigroups will lead to the same expanded algebra. Also, it is  shown that the
other Bianchi the type IV, VI-IX algebras cannot be obtained as an
expansion from the isometry algebras in $2$ dimensions. This means that the
first isometry algebras have properties that can be obtained from isometries
in $2$ dimensions but the second set have properties that are in some sense
intrinsic in $3$ dimensions.

The results obtained in this work are interesting, because even when $2$ and
$3$-dimensional isometry algebras are well known in the literature, the
\textit{non-trivial relations} we have found are something completely new. To
perform extensions of Lie algebras by adding generators and where the original
algebra is a subalgebra of the resulting ones it could be thought as a simple
problem. This is for example what happens with the algebras type III and V
(see table (\ref{Tab2}) and (\ref{Bianchi})) where they are simply two ways to
extend the algebra (\ref{2dimalg_dos}) to a 3 dimensional one (it is direct to
see that the algebra (\ref{2dimalg_dos}) is contained as a subalgebra in the
algebras type III and V of table (\ref{Bianchi})). In the present work we have
obtained these results by means of the expansion procedure\footnote{This in
some sense tell us that the $S$-expansion method includes not only the
contraction methods: contains also some (and probably all) kind of the
extensions procedures of a Lie algebra. A formal study about this theoretical
result is a work in progress (See ref. \cite{Poli}).} but, even
further, new kinds of relations were found. This is the case of the algebras
"type II" (see table (\ref{Tab2})) that was obtained as an expansion of the
algebra (\ref{2dimalg_dos}) and where this original algebra  is not present
as a subalgebra. This is why we  refer to these relations between
$2$ and $3$-dimensional as \textit{non-trivial} relations. These results are
completely new and can just be reached by means of the $S$-expansion
method\footnote{Note that the expansion method by using a parameter is
equivalent to an $S$-expansion but using just one specific semigroup, the
semigroup $S_{E}^{\left(  n\right)  }$ introduced in ref. \cite{irs}. Then
using another semigroups will lead us to more general expansions and that's
why effectively this results can just be obtained via the $S$-expansion
procedure.}.

It is also interesting that to face this problem we had to look for other
semigroups that have not been used yet in the applications of the
$S$-expansion method (the semigroups $S_{N1}$, $S_{N2}$, $S_{N3}$ and
$S_{K}^{\left(  3\right)  }$). Their principal properties, as the problem of
finding a resonant decomposition, were studied for each of them. By using
computing programs we have checked our results and solved the problem in a
 complete way.

In general, to understand whether two sets of algebras can be related by
means of an expansion is very interesting problem from both, physical and mathematical
point of view. In fact, many physical applications have been found in this
context: for example, in \cite{aipv1} the $M$-algebra is obtained as an
expansion of the $\mathfrak{osp}$(32/1) algebra. In fact, in \cite{irs1} this
result was re-obtained but via the $S$-expansion method which gives in
addition the invariant tensors of the expanded algebra. In this way, in the
mentioned reference, an eleven-dimensional gauge theory for the $M$-algebra
was constructed. Another interesting application is \cite{irs2} where
(2+1)-dimensional Chern-Simons AdS gravity is obtained from the so-called
"`exotic gravity"' and \cite{K15} where Standard General Relativity is
obtained from Chern-Simons Gravity. Finally, a generalization of the results
presented here can be useful to study isometries in higher dimensions,
particularly, in applications related with isometries of black holes solutions. A
first step on this direction is done in \cite{bianmetric}.

We conclude by remarking that what remains in common in all the physical
applications mentioned above is the question: \textit{given two symmetry
algebras, can they be related by means of some contraction or expansion
procedure?} The method presented in this paper is very instructive to answer on
this question. If the answer is \textit{yes}, there is a way to construct the
semigroup that gives the relation (as made in section \ref{Bianchi_related}).
On the contrary, if the answer is $no$ then it should be shown explicitly that
there exist no expansion method that can reproduce this relation (as made in
section \ref{Bianchi_no_related}). This mechanism can be developed to a
general algorithm to study more complicated cases, where the construction of
the semigroups by hand (as we made it here) would be impossible. General
criteria and the mentioned algorithm is a work in progress (see \cite{Poli}).

\section{\textbf{Acknoledgements}}

We are grateful to Patricio Salgado for many valuable discussions and for the
establishing the task. I.K. was supported by Fondecyt (Chile) grants 1050512,
1121030 and by DIUBB grant (UBB, Chile) 102609.  N. M. \& F.N. thank R.
D'Auria, M. Trigiante and L. Andrianopoli for their kind hospitality at
Dipartamento di Fisica of Politecnico di Torino, where part of this work was
done. N.M. is also grateful to the Comisi\'{o}n Nacional de Investigaci\'{o}n
Cient\'{\i}fica y Tecnol\'{o}gica Conicyt (Chile) for financial support
through a Becas-Chile grant. F.N. thanks CSIC for a JAE-predoc grant, 
cofunded by the European Social Fund.

\section{\textbf{Appendix}}

In this appendix we give explicitly the multiplication tables of the semigroups
that we have used in this paper and that belongs to the list generated by the
program \textit{com.f} of \cite{Plemmons}\ for $n=4$. Those semigroups are:%

\[%
\begin{tabular}
[c]{l|llll}%
$S_{\left(  4\right)  }^{10}$ & $\lambda_{1}$ & $\lambda_{2}$ & $\lambda_{3}$
& $\lambda_{4}$\\\hline
$\lambda_{1}$ & $\lambda_{1}$ & $\lambda_{1}$ & $\lambda_{1}$ & $\lambda_{1}%
$\\
$\lambda_{2}$ & $\lambda_{1}$ & $\lambda_{1}$ & $\lambda_{1}$ & $\lambda_{1}%
$\\
$\lambda_{3}$ & $\lambda_{1}$ & $\lambda_{1}$ & $\lambda_{1}$ & $\lambda_{2}%
$\\
$\lambda_{4}$ & $\lambda_{1}$ & $\lambda_{1}$ & $\lambda_{2}$ & $\lambda_{1}$%
\end{tabular}
\ \ \ \text{; \ }%
\begin{tabular}
[c]{l|llll}%
$S_{\left(  4\right)  }^{12}$ & $\lambda_{1}$ & $\lambda_{2}$ & $\lambda_{3}$
& $\lambda_{4}$\\\hline
$\lambda_{1}$ & $\lambda_{1}$ & $\lambda_{1}$ & $\lambda_{1}$ & $\lambda_{1}%
$\\
$\lambda_{2}$ & $\lambda_{1}$ & $\lambda_{1}$ & $\lambda_{1}$ & $\lambda_{1}%
$\\
$\lambda_{3}$ & $\lambda_{1}$ & $\lambda_{1}$ & $\lambda_{1}$ & $\lambda_{2}%
$\\
$\lambda_{4}$ & $\lambda_{1}$ & $\lambda_{1}$ & $\lambda_{2}$ & $\lambda_{3}$%
\end{tabular}
\ \ \ \text{;\ \ \ }%
\begin{tabular}
[c]{l|llll}%
$S_{\left(  4\right)  }^{13}$ & $\lambda_{1}$ & $\lambda_{2}$ & $\lambda_{3}$
& $\lambda_{4}$\\\hline
$\lambda_{1}$ & $\lambda_{1}$ & $\lambda_{1}$ & $\lambda_{1}$ & $\lambda_{1}%
$\\
$\lambda_{2}$ & $\lambda_{1}$ & $\lambda_{1}$ & $\lambda_{1}$ & $\lambda_{1}%
$\\
$\lambda_{3}$ & $\lambda_{1}$ & $\lambda_{1}$ & $\lambda_{1}$ & $\lambda_{3}%
$\\
$\lambda_{4}$ & $\lambda_{1}$ & $\lambda_{1}$ & $\lambda_{3}$ & $\lambda_{4}$%
\end{tabular}
\]

\[%
\begin{tabular}
[c]{l|llll}%
$S_{\left(  4\right)  }^{28}$ & $\lambda_{1}$ & $\lambda_{2}$ & $\lambda_{3}$
& $\lambda_{4}$\\\hline
$\lambda_{1}$ & $\lambda_{1}$ & $\lambda_{1}$ & $\lambda_{1}$ & $\lambda_{1}%
$\\
$\lambda_{2}$ & $\lambda_{1}$ & $\lambda_{1}$ & $\lambda_{1}$ & $\lambda_{1}%
$\\
$\lambda_{3}$ & $\lambda_{1}$ & $\lambda_{1}$ & $\lambda_{3}$ & $\lambda_{4}%
$\\
$\lambda_{4}$ & $\lambda_{1}$ & $\lambda_{1}$ & $\lambda_{4}$ & $\lambda_{3}$%
\end{tabular}
\ \ \ \text{; \ \ }%
\begin{tabular}
[c]{l|llll}%
$S_{\left(  4\right)  }^{42}$ & $\lambda_{1}$ & $\lambda_{2}$ & $\lambda_{3}$
& $\lambda_{4}$\\\hline
$\lambda_{1}$ & $\lambda_{1}$ & $\lambda_{1}$ & $\lambda_{1}$ & $\lambda_{1}%
$\\
$\lambda_{2}$ & $\lambda_{1}$ & $\lambda_{1}$ & $\lambda_{1}$ & $\lambda_{2}%
$\\
$\lambda_{3}$ & $\lambda_{1}$ & $\lambda_{1}$ & $\lambda_{1}$ & $\lambda_{3}%
$\\
$\lambda_{4}$ & $\lambda_{1}$ & $\lambda_{2}$ & $\lambda_{3}$ & $\lambda_{4}$%
\end{tabular}
\ \ \ \text{; \ \ }%
\begin{tabular}
[c]{l|llll}%
$S_{\left(  4\right)  }^{43}$ & $\lambda_{1}$ & $\lambda_{2}$ & $\lambda_{3}$
& $\lambda_{4}$\\\hline
$\lambda_{1}$ & $\lambda_{1}$ & $\lambda_{1}$ & $\lambda_{1}$ & $\lambda_{1}%
$\\
$\lambda_{2}$ & $\lambda_{1}$ & $\lambda_{1}$ & $\lambda_{1}$ & $\lambda_{2}%
$\\
$\lambda_{3}$ & $\lambda_{1}$ & $\lambda_{1}$ & $\lambda_{2}$ & $\lambda_{3}%
$\\
$\lambda_{4}$ & $\lambda_{1}$ & $\lambda_{2}$ & $\lambda_{3}$ & $\lambda_{4}$%
\end{tabular}
\]

\[%
\begin{tabular}
[c]{l|llll}%
$S_{\left(  4\right)  }^{44}$ & $\lambda_{1}$ & $\lambda_{2}$ & $\lambda_{3}$
& $\lambda_{4}$\\\hline
$\lambda_{1}$ & $\lambda_{1}$ & $\lambda_{1}$ & $\lambda_{1}$ & $\lambda_{1}%
$\\
$\lambda_{2}$ & $\lambda_{1}$ & $\lambda_{1}$ & $\lambda_{1}$ & $\lambda_{2}%
$\\
$\lambda_{3}$ & $\lambda_{1}$ & $\lambda_{1}$ & $\lambda_{3}$ & $\lambda_{1}%
$\\
$\lambda_{4}$ & $\lambda_{1}$ & $\lambda_{2}$ & $\lambda_{1}$ & $\lambda_{4}$%
\end{tabular}
\ \ \text{\ ; \ \ }%
\begin{tabular}
[c]{l|llll}%
$S_{\left(  4\right)  }^{45}$ & $\lambda_{1}$ & $\lambda_{2}$ & $\lambda_{3}$
& $\lambda_{4}$\\\hline
$\lambda_{1}$ & $\lambda_{1}$ & $\lambda_{1}$ & $\lambda_{1}$ & $\lambda_{1}%
$\\
$\lambda_{2}$ & $\lambda_{1}$ & $\lambda_{1}$ & $\lambda_{1}$ & $\lambda_{2}%
$\\
$\lambda_{3}$ & $\lambda_{1}$ & $\lambda_{1}$ & $\lambda_{3}$ & $\lambda_{3}%
$\\
$\lambda_{4}$ & $\lambda_{1}$ & $\lambda_{2}$ & $\lambda_{3}$ & $\lambda_{4}$%
\end{tabular}
\ \text{ \ ; \ \ }%
\begin{tabular}
[c]{l|llll}%
$S_{\left(  4\right)  }^{64}$ & $\lambda_{1}$ & $\lambda_{2}$ & $\lambda_{3}$
& $\lambda_{4}$\\\hline
$\lambda_{1}$ & $\lambda_{1}$ & $\lambda_{1}$ & $\lambda_{1}$ & $\lambda_{1}%
$\\
$\lambda_{2}$ & $\lambda_{1}$ & $\lambda_{2}$ & $\lambda_{1}$ & $\lambda_{1}%
$\\
$\lambda_{3}$ & $\lambda_{1}$ & $\lambda_{1}$ & $\lambda_{3}$ & $\lambda_{4}%
$\\
$\lambda_{4}$ & $\lambda_{1}$ & $\lambda_{1}$ & $\lambda_{4}$ & $\lambda_{3}$%
\end{tabular}
\]

\end{document}